\documentclass[lettersize,onecolumn]{IEEEtran}
\usepackage[caption=false,font=normalsize,labelfont=sf,textfont=sf]{subfig}
\usepackage{mymacro}
\usepackage{booktabs} 
\newcommand{\wt}{\mathrm{wt}}

\floatname{algorithm}{Protocol} % change title 'algorithm' to 'protocol'

\title{String commitment from unstructured noise}
\author[1]{Jiawei Wu}
\author[2,3,4]{Masahito Hayashi}
\author[1,5]{Marco Tomamichel}

\affil[1]{Centre for Quantum Technologies, National University of Singapore, Singapore}
\affil[2]{The Chinese University of Hong Kong, Shenzhen, Longgang District, Shenzhen, 518172, China}
\affil[3]{International Quantum Academy, Futian District, Shenzhen 518048, China}
\affil[4]{Graduate School of Mathematics, Nagoya University, Chikusa-ku, Nagoya 464-8602, Japan}
\affil[5]{Department of Electrical and Computer Engineering, National University of Singapore, Singapore}

\date{}

\begin{document}
\maketitle
\begin{abstract}
    Noisy channels are a foundational resource for constructing cryptographic primitives such as string commitment and oblivious transfer. The noisy channel model has been extended to unfair noisy channels, where adversaries can influence the parameters of a memoryless channel.
    In this work, we introduce the unstructured noisy channel model as a generalization of the unfair noisy channel model to allow the adversary to manipulate the channel arbitrarily subject to certain entropic constraints.
    We present a string commitment protocol with established security and derive its achievable commitment rate, demonstrating the feasibility of commitment against this stronger class of adversaries.
    Furthermore, we show that the entropic constraints in the unstructured noisy channel model can be derived from physical assumptions such as noisy quantum storage. 
    Our work thus connects two distinct approaches to commitment, i.e., the noisy channel and physical limitations.
\end{abstract}

\section{Introduction}
\subsection{Commitments}
Commitments are a typical two-party cryptographic primitive, serving as a critical component for various cryptographic applications, such as coin flipping~\cite{blum1983coin,demay2013unfair}, zero-knowledge proofs~\cite{brassard1988minimum,goldreich1991proofs}, and secure multiparty computation~\cite{crepeau2020commitment,ben-or2006secure,dupuis2010secure,goldreich2019how}. 
Commitments allow a party, Alice, to commit to a specific value (usually a bit or bit string) in a way that keeps the value hidden from another party, Bob, until a later point when Alice chooses to reveal the committed value. The two critical properties of commitments are the hiding and the binding properties.
(1) The hiding property ensures that Bob is ignorant of the committed value during the commit phase.
(2) The binding property ensures that, once the commitment has been established, Alice is precluded from altering the value she intends to disclose.
A common analogy for commitment involves Alice locking a message inside a container and sending it to Bob. 
At this stage, Bob remains unaware of the actual content of the message. 
Later, upon Alice's provision of the corresponding key, Bob can unlock the container and verify the committed value.

The study of commitment traces back to Blum's foundational work~\cite{blum1983coin}, where commitment is used to implement coin flipping and is shown to be secure under the assumption that factoring is hard. 
In classical settings, commitments can be achieved under certain computational assumptions, such as existence of one-way functions, given that the adversary's computational power is limited.
While research based on these two types of assumptions has been fruitful, 
the success of quantum key distribution suggests that the computational assumptions can be removed for some cryptographic tasks.
However, even with quantum communication, common two-party cryptographic tasks, including commitment and secure function evaluation, have been proven to be impossible~\cite{lo1997quantum,mayers1997unconditionally,lo1998why,winkler2011commitment}.

Research has since considered physical assumptions that make commitment possible, such as special relativity~\cite{kent1999relativistic,croke2012relativistic,kaniewski2013relativistic,lunghi2013relativistic,lunghi2015relativistic}, bounded quantum storage~\cite{damgard2008cryptography,wehner2008composable,barhoush2023powerful} and noisy quantum storage~\cite{wehner2008cryptography,wehner2008cryptographyb,schaffnerchristian2009robust,konig2012unconditional}.
On the other hand, a parallel line of work has focused on constructing commitment from abstract noisy channels. These two approaches—one based on concrete physical limitations and the other on abstract channel properties—have largely evolved independently.
A natural question is whether a more general channel model can unify these two fundamentally different types of constraints. This is the primary motivation for the present work.

\subsection{The Noisy Channel Model}
The inherent noise of channels, which is independent of adversarial influence, has been noticed to be useful for cryptographic tasks.
Wyner's wiretap channel model~\cite{wyner1975wiretap} and its generalization~\cite{csiszar1978broadcast} utilize the noisy gap between two channels to implement secure communication in the presence of an eavesdropper.
More recent works demonstrate that noisy channels can support various two-party cryptographic primitives~\cite{crepeau1997efficient}, including bit commitment and oblivious transfer.

From an information-theoretical viewpoint, it is important to determine the optimal rate of a primitive constructed from a channel.
Specifically, for string commitment, the commitment rate is defined as the ratio between the number of bits committed and the time of channel use (see equation \eqref{eq:ratedef} for a formal definition).
Both string commitment capacity~\cite{winter2003commitment,crepeau2005efficient,hayashi2022commitment,hayashi2023commitment} and string oblivious transfer capacity~\cite{crepeau2005efficient,ishai2011constantrate,dowsley2017oblivious} have been extensively studied. For example, the commitment capacity of a binary symmetric channel with transition probability $p$ is $C = h(p)$, where $h(\cdot)$ is the binary entropy function.

In more realistic scenarios, adversaries may have partial control over the channel, potentially influencing its error rate. 
This inspired the study of two-party cryptographic primitives based on unfair noisy channel model~\cite{damgard1999im,cascudo2016oblivious,khurana2016secure,crepeau2020commitment}. 
Specifically, unfair noisy channel model considers the following setting: 
\begin{itemize}
    \item If Alice and Bob are both honest, then the channel is a binary symmetric channel with transition probability $\delta$.
    \item If one party, either Alice or Bob, is dishonest, then she (he) can configure the transition probability as any value in interval $[\gamma, \delta]$, while keeping the honest party ignorant about this value.
\end{itemize}
Compared with a fixed binary symmetric channel, this model gives the dishonest party the option to reduce channel error, which enhances the cheating ability. 
Cr\'{e}peau et al.\cite{crepeau2020commitment} derived the commitment capacity of the unfair noisy channel model: $C = h(\delta) - h(\frac{\delta -\gamma}{1-2\gamma})$.

A key limitation of the unfair noisy channel model,  is the i.i.d. assumption even in the dishonest case. Specifically, a malicious adversary with control over the physical medium may not be restricted to simply tuning a single error parameter. Rather, he could introduce memory effects to maximize his advantage. The assumption of a memoryless adversarial channel is a mathematical convenience but may not hold in realistic scenarios. This motivates the development of a model that captures such non-i.i.d. behavior.

\subsection{Contributions}
Our contributions are two-fold.
First, we introduce and formalize the unstructured noisy channel (UsNC) model for commitment task, which generalizes prior unfair noisy channel models by removing the i.i.d. assumption of the adversarial channel.
In the UsNC model, the channel remains memoryless when both Alice and Bob are honest, while a dishonest party may introduce arbitrary memory effects, constrained by global entropic bounds.
Within this model, we propose a string commitment protocol (Protocol~\ref{pt:com}) inspired by~\cite{winter2003commitment,imai2006efficient},
and formally prove its security by establishing its completeness, hiding, and binding properties (Theorem~\ref{thm:main}).

Second, we demonstrate that the UsNC model can be directly instantiated by the concrete physical constraints of the Noisy-Quantum-Storage (NQS) model.
In Section~\ref{sec:Application}, we establish the conditions of the UsNC model from noisy quantum storage, thereby showing an alternative commitment scheme from NQS model other than~\cite{konig2012unconditional}.
This connection successfully bridges two important, but largely separate lines of research in two-party cryptography.

\section{Preliminaries}
In this section, we introduce some necessary notations, entropy quantities and some useful lemmas on information theory. 
For some definitions of classical version that can be easily derived from their quantum counterpart (like the generalized trace distance), we will only present the quantum version. Table~\ref{tab:notations} is a summary of selected notations and symbols.

\subsection{Basic notations}
Let $\cB(\cH)$ be the set of bounded operators on Hilbert space $\cH$. A (quantum) state is a positive semidefinite operator $\rho \in \cB(\cH)$ with normalized trace value $\tr(\rho) =0$. The set of all states on $\cH$ is denoted by $\cD(\cH)$. $\rho$ is called a subnormalized state if $0 \le \rho \le 1$.
The set of all subnormalized states on $\cH$ is denoted by $\cD_{\bullet}(\cH)$. 
A classical state of random variable $X$ can be expressed as $\rho_X = \sum_{x \in \cX} P_{X}(x) \proj{x}$. The joint state of a classical system $X$ and a quantum system $B$ is called a cq state, which takes the form $\rho_{XB} = \sum_{x} P_{X}(x) \proj{x} \otimes \rho_{B|x} $.

For distance measure on subnormalized states, we use generalized trace distance, defined as
\begin{align}
    \GTD(\rho,\sigma) = \frac{1}{2} \| \rho - \sigma \|_1 + \frac{1}{2} \left| \tr(\rho) - \tr(\sigma) \right|. \notag
\end{align}
It is shown in~\cite[Section 3.1.2]{tomamichel2015quantum} that $\GTD(\cdot,\cdot)$ is a metric. If $P,Q$ are normalized, then $\GTD(P,Q)$ reduces to the normal total variational distance.

For a finite set $\cX$, let $\Delta(\cX)$ be the probability simplex and $\Delta_{\bullet}(\cX)$ be the subnormalized simplex. They are defined as
\begin{align*}
    \Delta(\cX) \coloneqq \left\{ P:\cX \to [0,1] \middle| \sum_{x\in \cX} P(x) = 1 \right\}, \\
    \Delta_{\bullet}(\cX) \coloneqq \left\{ P:\cX \to [0,1] \middle| \sum_{x\in \cX} P(x) \le 1 \right\}.
\end{align*}
The elements in $\Delta_{\bullet}$ are called subnormalized distribution.
The uniform distribution on set $\cX$ is denoted by $U_{\cX}$.

For two discrete probability distributions $P,Q$ on set $\cX$, 
their generalized trace distance reduces to
\begin{equation}
\begin{aligned} \label{eq:deftd}
    \GTD(P,Q) &\coloneqq \frac{1}{2} \| P-Q\|_1 + 
    \frac{1}{2} \left| \| P \|_1 -\|Q\|_1 \right| \\
    &= \frac{1}{2} \sum_{x\in \cX} |P(x)-Q(x)| + \frac{1}{2}\left|\sum_{x \in \cX} P(x)-Q(x) \right|.
\end{aligned}
\end{equation}

A classical-quantum channel $W$ (cq channel for short) maps a finite set $\cX$ to a quantum state $W(x)$. 
% If $W(x) \in \cD_{\bullet}(\cH_{B})$, then the cq channel $W$ is called subnormalized.
For random variable $X$ subjected to distribution $P_X$, $W \times P_X$ denotes the joint state $\rho_{XB} = \sum_{x} P_{X}(x)\proj{x}\otimes \rho_{B|x}$.
and $W \circ P_X$ denotes the marginal output state on $B$, that is, $\sum_{x} P_X(x) \rho_{B|x}$.

\subsection{Entropic quantities}

We denote the von Neumann entropy (or Shannon entropy) of state $\rho_A \in \cD_{\bullet}(\cH_A)$ as $H(A)_\rho$.
The min-entropy of $\rho_A \in \cD_{\bullet}(\cH_A)$ and conditional min-entropy of $\rho_{AB} \in \cD_{\bullet}(\cH_{AB})$ are defined as
\begin{align}
    H_{\min}(A)_\rho &\coloneqq -D_{\infty}(\rho_A \| \mathbb{I}_A), \label{eq:q-min}\\
    H_{\min}(A|B)_\rho &\coloneqq  \max_{\sigma_B \in \cD(\cH_B)}  -D_{\infty}(\rho_{AB}\| \mathbb{I}_A \otimes \sigma_B), \label{eq:q-cmin}
\end{align}
where
$D_{\infty}(\rho \| \sigma) \coloneqq \inf \{ \lambda \in \mathbb{R}: \rho \le 2^\lambda \sigma\}$.

As a special case, for classical distributions $P_X$ and $P_{XY}$, equations~\eqref{eq:q-min} and \eqref{eq:q-cmin} simplify to
\begin{align}
    H_{\min}(X)_P & \coloneqq -\log \max_{x \in \cX} P_X(x), \label{eq:defmin}\\
    H_{\min}(X|Z)_P & \coloneqq - \log \sum_{z \in \ca{Z}}\max_{x \in \ca{X}}P_{X,Z}(x,z), \label{eq:def-c}
\end{align}
where $\log$ is taken with base $2$.

Accordingly, the smooth min-entropy of $\rho_A$ and smooth conditional entropy of $\rho_{AB}$ are respectively defined as 
\begin{align*}
    H_{\min}^{\varepsilon}(A)_\rho &\coloneqq  
    \max_{\substack{\rho_{A}' \in \cD_{\bullet}(\cH_{A})  \\
    \GTD(\rho_{A}',\rho_{A}) \le \varepsilon}} H_{\min}(A)_{\rho'},  \\
    H_{\min}^{\varepsilon}(A|B)_\rho &\coloneqq
    \max_{\substack{\rho_{AB}' \in \cD_{\bullet}(\cH_{AB}) \\
    \GTD(\rho_{AB}',\rho_{AB})\le \varepsilon}} H_{\min}(A|B)_{\rho'}.
\end{align*}

A classical-quantum channel $W$ (cq channel for short) maps a finite set $\cX$ to a quantum state $W(x)$. 
For random variable $X$ subjected to distribution $P_X$, $W \times P_X$ denotes the joint state $\rho_{XB} = \sum_{x} P_{X}(x)\proj{x}\otimes \rho_{B|x}$,
and $W \circ P_X$ denotes the marginal output state on $B$, that is, $\sum_{x} P_X(x) \rho_{B|x}$.

\subsection{Typical set}
We also need the concept of (weakly) conditional typical set. 
The sequence $x_1 x_2 \dots x_n$ is denoted by $\bx$ whenever the length $n$ is obvious from the context. 
For an input sequence $\bx \in \{0,1\}^n$ and binary symmetric channel (BSC) $W$, the  (weakly) conditional typical set is defined\footnote{The standard definition of weakly conditional typical set reduces to this Hamming distance form for binary symmetric channel.} as
\begin{align} \label{eq:defsct}
    \cT_{W,\varepsilon}^n (\bx) \coloneqq \left\{ \bz \in \{0,1\}^n \mid n(p-\varepsilon) \le \HD(\bx,\bz) \le n(p+\varepsilon) \right\},
\end{align}
where $\HD(\cdot,\cdot)$ is Hamming distance.
Because $\HD(\bx, \bz) = \HD(\bx + \by, \bz +\by)$, we have the following translation invariance property
\begin{align}
    \bz \in \cT_{W,\varepsilon}^n (\bx) \iff \bz + \by \in \cT_{W,\varepsilon}^n (\bx + \by). \label{eq:SYM}
\end{align}
The following lemma suggests that the output sequence of channel $W$ stays in the conditional typical set with large probability.
\begin{lemma}[\textnormal{\cite[inequality (13)]{winter2003commitment}}, adapted] \label{lem:unit}
Let $W$ be a binary symmetric channel. For any $\bx \in \{0,1\}^n$ and $\varepsilon >0$,
\begin{align*}
    \Pr_{\bZ \sim W^n(\bx)}[\bZ \in \cT_{W, \varepsilon}^n(\bx)] \ge 1 - 8 \cdot 2^{- n \varepsilon^2}.
\end{align*}
\end{lemma} 

\begin{table}[t!]
\caption{Selected notations and symbols}
\label{tab:notations}
\centering

\renewcommand{\arraystretch}{1.3}
\begin{tabular}{@{}l p{6.4cm} l@{}}
\toprule
\textbf{Symbol} & \textbf{Meaning} \\
\midrule
$\cB(\cH)$ & Set of bounded operators on $\cH$ \\
$\cD(\cH)$  & Set of density operators on $\cH$ \\
$\cD_{\bullet}(\cH)$  & Set of subnormalized density operators on $\cH$ \\
$\GTD(\cdot, \cdot)$  & Generalized Trace Distance     \\
$H_{\min}(A)_\rho$        & Min-entropy of system $A$                 \\
$H_{\min}^{\varepsilon}(A|B)_\rho$ & Smooth conditional min-entropy  \\
$\HD(\cdot, \cdot)$ & Hamming distance  \\
$\cT_{W,\varepsilon}^n(\bx)$ & (Weakly) Conditional typical set \\
$\theta$                  & UsNC parameters $(p, \varepsilon_A, l_A, \varepsilon_B, l_B)$  \\
$W^n$                     & Honest channel ($n$-fold BSC)            \\
$W^A, l_A$                & Dishonest Alice's channel \& entropy bound \\
$W^B, l_B$                & Dishonest Bob's channel \& entropy bound \\
$\cC$          & Linear code in $\{0,1\}^n$  \\
$\cC'$          & Coset of linear code $\cC$  \\
\bottomrule
\end{tabular}
\end{table}

\section{String Commitment from Unstructured Noisy Channel Model}
In this section, we first introduce a formal definition of the UsNC model and the requirements of string commitment in Section \ref{sec:TD}. Then we describe our protocol to construct string commitment from the UsNC model in Section \ref{sec:PBC}. 

\subsection{Task Definition} \label{sec:TD}
In two-party cryptography, it is standard practice to consider three adversarial scenarios: (1) both Alice and Bob are honest; (2) only Alice is dishonest; (3) only Bob is dishonest.
Accordingly, we define the UsNC model in Definition \ref{def:UsNC} to capture all three cases.
Throughout, we always assume a bidirectional noiseless channel is available. 

\begin{definition}[The unstructured noisy channel model] \label{def:UsNC}
    Let $p\in (0,\frac{1}{2}), n \in \mathbb{N}$. $\varepsilon_A, l_A, \varepsilon_B, l_B$ are non-negative real numbers, and $\theta:=(p,\varepsilon_A ,l_A,\varepsilon_B,l_B)$.
    An $(n,\theta)$-UsNC is characterized by the following conditions:
    \begin{description}
    \item[(C1)] If both Alice and Bob are honest, the channel from Alice to Bob is an $n$-fold binary symmetric channel (BSC) $W^n$ with transition probability $p$, where $W: \ca{X} \to \ca{Z}$ and $\cX=\cZ=\{0,1\}$.
    \item[(C2)] If Alice is dishonest, she can manipulate the channel to $W^A: \hat{\ca{X}} \to \cZ^n$ under the constraint: 
    \begin{align} \label{eq:ch2}
        \forall x \in \hat{\cX}, H_{\min}^{\varepsilon_A}(\bZ)_{ W^A(x)} \ge l_A,
    \end{align} 
    where $\hat{\cX}$ is an arbitrary set.
    \item[(C3)] If Bob is dishonest, he can manipulate the channel to a cq channel $W^B: \cX^n \to \cD(\cH_{B})$ under the constraint: 
    \begin{align} \label{eq:ch3}
    H_{\min}^{\varepsilon_B}(\bX|B)_{W^B \times U_{\cX^n}} \ge l_B,
    \end{align}
    where $B$ is an arbitrary quantum system.
\end{description}
\end{definition}

\begin{remark}
In condition \textbf{(C2)}, allowing an arbitrary input set models the fact that Alice has full control over the channel input.
The entropy constraint \eqref{eq:ch2} prevents the output distribution from being too concentrated.
In condition \textbf{(C3)}, the arbitrary quantum system $B$ models the fact that Bob has full control over the channel output.
The entropy constraint \eqref{eq:ch3} guarantees that the channel output does not reveal too much information about Alice's input.
Notably, the two conditions are different in terms of entropic quantities. 
In Section \ref{sec:proofBC}, we will see that this difference comes from the fundamental asymmetry in requirements for binding and hiding.
\end{remark}

Next, we formally define a string commitment protocol in the $(n,\theta)$-UsNC model.
String commitment protocol consists of a commit phase and a reveal phase.
Each phase involves certain rounds of interactions, while the commit phase includes additional use of the noisy channel.
We represent a general string commitment protocol with interactive random systems $\alpha_C$, $\alpha_R$, $\beta_C$,  $\beta_R$\footnote{These interactive systems can be further characterized by the quantum comb model~\cite{gutoski2007toward,chiribella2009theoretical}}, as shown in Figure~\ref{fig:bc}(a). 
In this representation, a message $M \in \cM$ is committed through the interaction between $\alpha_C$ and $\beta_C$.
In the reveal phase, honest Bob's system $\beta_R$ outputs the revealed message $\hat{M}$ and a flag bit $F \in \{\text{``acc''},\text{``rej''} \}$ representing accepting or rejecting the message.

\begin{figure}
    \centering
    \input{pic/general.tex}
    \caption{
    (a) The honest protocol is modelled by random interactive systems. $\alpha_C,\alpha_R$ are Alice's protocol in the commit phase and reveal phase. 
    $\beta_C,\beta_R$ are Bob's protocol in the commit phase and reveal phase. 
    The arrow $\nleftrightarrow$ denotes multi-rounds of interactions through the noiseless channel. The order of the use of noiseless channel and the use of noisy channel $W$ can be arbitrary. 
    Variables $R_\alpha, R_\beta$ denote all registers of system $\alpha_C,\beta_C$, respectively. $\hat{M}$ is the message inferred by $\beta_R$. $F \in \{\textrm{``acc''},\textrm{``rej''} \}$ is a flag representing the acceptance or rejection of the revealed message.
    (b) An illustration of dishonest Bob in the commit phase. $T_C$ is the Bob's transcript after interacting with Alice through $\cB_C$.
    (c) An illustration of dishonest Alice. $\cA_{R0}, \cA_{R1}$ and $\cA_C$ are interactive systems representing Alice's strategy. $\beta_{R0}$ and $\beta_{R1}$ are two copies of $\beta_R$.}
    \label{fig:bc}
\end{figure}

The requirements for string commitment are:
\begin{enumerate}
    \item Completeness: When both Alice and Bob are honest, the protocol proceeds as in Figure~\ref{fig:bc}(a). The completeness parameter $\delta_c$ is defined as
    \begin{align*}
    \delta_c:=
    \max_{m\in \ca{M}} 1- \Pr[\hat{M}= m \land F = \textrm{``acc''}].
    \end{align*}
    \item Hiding: When Bob is dishonest, denote his cheating strategy at commit phase as $\cB_C$.
    As shown in Figure~\ref{fig:bc}(b), $\cB_C$ is a probabilistic system that interacts with honest Alice's system $\alpha_C$. Let $B'$ be all the information of Bob at the end of commit phase when  Bob applies strategy $\cB_C$. The hiding parameter $\delta_h$ is defined as
    \begin{align*}
    \delta_h:=
        \max_{\cB_C}
        \max_{m\ne m'\in \ca{M}}
        \GTD(\rho_{B|m}, \rho_{B|m'}).  % \label{eq:HD}
    \end{align*}
    \item Binding: This condition corresponds to soundness from Bob's viewpoint. When Alice is dishonest, she applies the cheating strategy  $\cA_C$ to interact with $\beta_C$ in the commit phase and two strategies $\cA_{R0},\cA_{R1}$ to interact with $\beta_{R0}, \beta_{R1}$ respectively in the reveal phase, as shown in Figure~\ref{fig:bc}(c). 
    Alice succeeds when two different messages are accepted by Bob.
    Thus, the binding parameter $\delta_b$ is defined as
    \begin{align*}
    \delta_b:=
    \max_{\cA_C,\cA_{R0}, \cA_{R1}}
    \Pr
    \left[
        \begin{array}{l}
        F_0 = F_1 = \text{``acc'' } \land \\
        \hat{M}_0 \ne \hat{M}_1 
        \end{array}
    \right], %\label{eq:binding}
    \end{align*}
    where $F_0,F_1,\hat{M}_0,\hat{M}_1$ are the output variables of $\beta_{R0}$ and $\beta_{R1}$.
\end{enumerate}

For commitment protocol $\Phi$ built from $(n,\theta)$-UsNC, we write $\delta_c,\delta_h,\delta_b$ as
$\delta_c(n,\theta,\Phi)$,
$\delta_h(n,\theta,\Phi)$,
$\delta_b(n,\theta,\Phi)$.

In this section, we aim to find a protocol $\Phi$
such that
$\delta_c(n,\theta,\Phi)$,
$\delta_h(n,\theta,\Phi)$, and
$\delta_b(n,\theta,\Phi)$ are sufficiently small. The string commitment rate is defined as 
\begin{align} \label{eq:ratedef}
    R \coloneqq \frac{\log |\cM|}{n}.
\end{align}

\begin{definition}[Achievable commitment in the UsNC model] \label{def:crate}
  For a sequence of UsNC models $\{(n,\theta_n)\}_{n\in \mathbb{N}}$, a rate $R$ is said to be achievable  if, for any $R'<R$, there exists a sequence of protocol $\{\Phi_n \}_{n \in \mathbb{N}}$ such that $\Phi_n$ constructs string commitment on message set $\cM_n$ and 
\begin{align*}
    \lim_{n \to \infty} \delta_{c}(n,\theta_n,\Phi_n) &= 0,\\
    \lim_{n \to \infty} \delta_{h}(n,\theta_n,\Phi_n) &=0,\\
    \lim_{n \to \infty} \delta_{b}(n,\theta_n,\Phi_n) &= 0,\\
    \lim_{n \to \infty} \frac{\log |\cM_n|}{n} &= R'.
\end{align*}
  
\end{definition}

\subsection{Protocol} \label{sec:PBC}

In this section we will describe our string commitment protocol and our main result on its completeness, hiding property and binding property.

Suppose $\varepsilon>0$ is an arbitrarily small number, 
$\cM$ is the message set with an additive group structure.
$\cC \subset \{0,1\}^n$ is a linear code with Hamming distance $\HD(\cC) = d < n/2$, where $\HD(\cC) \coloneqq \min_{\bx \ne \by \in \cC} \HD(\bx ,\by)$.
A coset $\cC' \in \cX^n / \cC$ can be constructed by selecting an arbitrary element $\bx_{\cC'} \in \cC'$ so that $\cC' = \cC + \bx_{\cC'}$. For each $\cC'$, we fix the element $\bx_{\cC'}$ and call it the representative element of $\cC'$.

Let $F=\{f_s\}_{s\in \ca{S}} : \cC \to \ca{M}$ be a balanced 2-universal hash function (UHF) family\footnote{The function family $\{f_s\}_{s\in \cS} : \cC \to \cM$ is a balanced 2-universal hash function family if (1) $\forall m, m' \in \cM, s\in \cS, |f_s^{-1}(m)|=|f_s^{-1}(m')|$ and (2) $\forall c_1 \ne c_2 \in \cC, \Pr[f_S(c_1) = f_S(c_2)] \le 1/|\cM|$ when $S$ is uniformly distributed on $\cS$.}.
Define the random map $\Gamma_s: \cM \to \cC$ as 
\begin{align*}
    \Gamma_s (c|m) = \frac{|\cM|}{|\cC|} \mathbbm{1}[f_s(c)=m],
\end{align*}
where $\mathbbm{1}[\cdot]$ is the indicator function.

With the above preparation, we propose Protocol \ref{pt:com}, denoted by
$\Phi[\varepsilon,\cC,F]$. 
The protocol is visualized in Figure~\ref{fig:protocol}
The claim about this protocol is shown in Theorem~\ref{thm:main}.

\begin{algorithm}[h!]
\caption{String commitment protocol from UsNC} \label{pt:com}
\vspace{1ex}
\textbf{Commit phase}:
\begin{itemize}
    \item To commit to message $M \in \cM$, Alice first generates uniformly random variables $S \in_R \cS$, $\oM \in_R \cM$ and $\cC' \in \cX^n / \cC$. Then she computes $\uM = M + \oM$, applies the random map $\Gamma_S$ to $\uM$ to get $\bX$, and sends $\uX = \bX+\bx_{\cC'}$ through the noisy channel. 
    Later, Alice sends $S$, $\oM$, $\cC'$ via a noiseless channel.
    \item Bob keeps $S,\oM, \cC'$ and the channel output $\bZ$ as the commitment.
\end{itemize}
\textbf{Reveal phase}:
\begin{itemize}
    \item Alice announces $(M,\bX)$. Bob outputs $\hat{M}=M$ and  tests if $\bX \in \cC$, $\bZ \in \cT_{W,\varepsilon}^n(\bX + \bx_{\cC'})$, $f_S(\bX) = M + \oM$. If the test passes, Bob outputs $F = \textrm{``acc''}$, else $F = \textrm{``rej''}$.
\end{itemize}
\end{algorithm}

\begin{figure} 
    \centering
    \tikzset{every picture/.style={line width=0.75pt}} %set default line width to 0.75pt        

\begin{tikzpicture}[x=0.75pt,y=0.75pt,yscale=-1,xscale=1]
%uncomment if require: \path (0,390); %set diagram left start at 0, and has height of 390

%Shape: Rectangle [id:dp24729185322921532] 
\draw  [dash pattern={on 0.75pt off 0.75pt}] (356.33,97.71) -- (426.33,97.71) -- (426.33,209.04) -- (356.33,209.04) -- cycle ;
%Straight Lines [id:da6526820357508085] 
\draw    (219,130.33) -- (293.67,130.33) ;
\draw [shift={(295.67,130.33)}, rotate = 180] [color={rgb, 255:red, 0; green, 0; blue, 0 }  ][line width=0.75]    (10.93,-3.29) .. controls (6.95,-1.4) and (3.31,-0.3) .. (0,0) .. controls (3.31,0.3) and (6.95,1.4) .. (10.93,3.29)   ;
%Shape: Rectangle [id:dp9673439946144116] 
\draw   (295.38,115.59) -- (325.41,115.59) -- (325.41,144.92) -- (295.38,144.92) -- cycle ;

%Straight Lines [id:da8216489043877095] 
\draw    (325.38,130.33) -- (396.67,130.33) ;
\draw [shift={(398.67,130.33)}, rotate = 180] [color={rgb, 255:red, 0; green, 0; blue, 0 }  ][line width=0.75]    (10.93,-3.29) .. controls (6.95,-1.4) and (3.31,-0.3) .. (0,0) .. controls (3.31,0.3) and (6.95,1.4) .. (10.93,3.29)   ;
%Straight Lines [id:da3049968681995583] 
\draw    (183.01,270.33) -- (355.67,270.33) ;
\draw [shift={(357.67,270.33)}, rotate = 180] [color={rgb, 255:red, 0; green, 0; blue, 0 }  ][line width=0.75]    (10.93,-3.29) .. controls (6.95,-1.4) and (3.31,-0.3) .. (0,0) .. controls (3.31,0.3) and (6.95,1.4) .. (10.93,3.29)   ;
%Straight Lines [id:da6842939056856735] 
\draw    (426.04,256) -- (464.04,256) ;
\draw [shift={(466.04,256)}, rotate = 180] [color={rgb, 255:red, 0; green, 0; blue, 0 }  ][line width=0.75]    (10.93,-3.29) .. controls (6.95,-1.4) and (3.31,-0.3) .. (0,0) .. controls (3.31,0.3) and (6.95,1.4) .. (10.93,3.29)   ;
%Shape: Rectangle [id:dp5890116152578648] 
\draw  [dash pattern={on 0.75pt off 0.75pt}] (134.33,97.71) -- (266,97.71) -- (266,212.38) -- (134.33,212.38) -- cycle ;
%Straight Lines [id:da8475799416746003] 
\draw    (109.67,130) -- (132,130) ;
\draw [shift={(134,130)}, rotate = 180] [color={rgb, 255:red, 0; green, 0; blue, 0 }  ][line width=0.75]    (10.93,-3.29) .. controls (6.95,-1.4) and (3.31,-0.3) .. (0,0) .. controls (3.31,0.3) and (6.95,1.4) .. (10.93,3.29)   ;
%Rounded Rect [id:dp3672432230317162] 
\draw   (184.67,111.33) .. controls (184.67,111.33) and (184.67,111.33) .. (184.67,111.33) -- (218.44,111.33) .. controls (218.44,111.33) and (218.44,111.33) .. (218.44,111.33) -- (218.44,145.71) .. controls (218.44,145.71) and (218.44,145.71) .. (218.44,145.71) -- (184.67,145.71) .. controls (184.67,145.71) and (184.67,145.71) .. (184.67,145.71) -- cycle ;
%Rounded Rect [id:dp8411588542797502] 
\draw   (185.33,164.26) .. controls (185.33,160.52) and (188.36,157.5) .. (192.09,157.5) -- (212.36,157.5) .. controls (216.09,157.5) and (219.11,160.52) .. (219.11,164.26) -- (219.11,185.12) .. controls (219.11,188.85) and (216.09,191.88) .. (212.36,191.88) -- (192.09,191.88) .. controls (188.36,191.88) and (185.33,188.85) .. (185.33,185.12) -- cycle ;

%Straight Lines [id:da30560604217428167] 
\draw    (201.29,158.11) -- (201.29,147.54) ;
\draw [shift={(201.29,145.54)}, rotate = 90] [color={rgb, 255:red, 0; green, 0; blue, 0 }  ][line width=0.75]    (10.93,-3.29) .. controls (6.95,-1.4) and (3.31,-0.3) .. (0,0) .. controls (3.31,0.3) and (6.95,1.4) .. (10.93,3.29)   ;
%Rounded Rect [id:dp030096648871597242] 
\draw  [dash pattern={on 0.75pt off 0.75pt}] (358.18,238.38) .. controls (358.18,238.38) and (358.18,238.38) .. (358.18,238.38) -- (426.18,238.38) .. controls (426.18,238.38) and (426.18,238.38) .. (426.18,238.38) -- (426.18,299.04) .. controls (426.18,299.04) and (426.18,299.04) .. (426.18,299.04) -- (358.18,299.04) .. controls (358.18,299.04) and (358.18,299.04) .. (358.18,299.04) -- cycle ;
%Rounded Rect [id:dp9206138402425054] 
\draw   (138.33,164.26) .. controls (138.33,160.52) and (141.36,157.5) .. (145.09,157.5) -- (165.36,157.5) .. controls (169.09,157.5) and (172.11,160.52) .. (172.11,164.26) -- (172.11,185.12) .. controls (172.11,188.85) and (169.09,191.88) .. (165.36,191.88) -- (145.09,191.88) .. controls (141.36,191.88) and (138.33,188.85) .. (138.33,185.12) -- cycle ;

\draw   (144.67,129.57) .. controls (144.67,124.17) and (149.04,119.8) .. (154.43,119.8) .. controls (159.83,119.8) and (164.2,124.17) .. (164.2,129.57) .. controls (164.2,134.96) and (159.83,139.33) .. (154.43,139.33) .. controls (149.04,139.33) and (144.67,134.96) .. (144.67,129.57) -- cycle ; \draw   (144.67,129.57) -- (164.2,129.57) ; \draw   (154.43,119.8) -- (154.43,139.33) ;
%Straight Lines [id:da61117987567796] 
\draw    (164.71,129.67) -- (181.67,129.67) ;
\draw [shift={(183.67,129.67)}, rotate = 180] [color={rgb, 255:red, 0; green, 0; blue, 0 }  ][line width=0.75]    (10.93,-3.29) .. controls (6.95,-1.4) and (3.31,-0.3) .. (0,0) .. controls (3.31,0.3) and (6.95,1.4) .. (10.93,3.29)   ;
%Straight Lines [id:da9194943635370774] 
\draw    (154.29,156.37) -- (154.29,142.21) ;
\draw [shift={(154.29,140.21)}, rotate = 90] [color={rgb, 255:red, 0; green, 0; blue, 0 }  ][line width=0.75]    (10.93,-3.29) .. controls (6.95,-1.4) and (3.31,-0.3) .. (0,0) .. controls (3.31,0.3) and (6.95,1.4) .. (10.93,3.29)   ;
%Straight Lines [id:da25077443702646196] 
\draw    (265,175.71) -- (388.24,175.71) ;
\draw [shift={(390.24,175.71)}, rotate = 180] [color={rgb, 255:red, 0; green, 0; blue, 0 }  ][line width=0.75]    (10.93,-3.29) .. controls (6.95,-1.4) and (3.31,-0.3) .. (0,0) .. controls (3.31,0.3) and (6.95,1.4) .. (10.93,3.29)   ;
%Straight Lines [id:da6787027186925709] 
\draw    (134.33,130) -- (145.24,130) ;
%Straight Lines [id:da5598687681108598] 
\draw    (390.24,175.71) -- (390.24,235.71) ;
\draw [shift={(390.24,237.71)}, rotate = 270] [color={rgb, 255:red, 0; green, 0; blue, 0 }  ][line width=0.75]    (10.93,-3.29) .. controls (6.95,-1.4) and (3.31,-0.3) .. (0,0) .. controls (3.31,0.3) and (6.95,1.4) .. (10.93,3.29)   ;
%Straight Lines [id:da29792574081000867] 
\draw    (398.67,130.33) -- (398.67,235.04) ;
\draw [shift={(398.67,237.04)}, rotate = 270] [color={rgb, 255:red, 0; green, 0; blue, 0 }  ][line width=0.75]    (10.93,-3.29) .. controls (6.95,-1.4) and (3.31,-0.3) .. (0,0) .. controls (3.31,0.3) and (6.95,1.4) .. (10.93,3.29)   ;
%Straight Lines [id:da46804264753315916] 
\draw    (183.01,268.38) -- (183.01,211.38) ;
\draw [shift={(183.01,270.38)}, rotate = 270] [color={rgb, 255:red, 0; green, 0; blue, 0 }  ][line width=0.75]    (10.93,-3.29) .. controls (6.95,-1.4) and (3.31,-0.3) .. (0,0) .. controls (3.31,0.3) and (6.95,1.4) .. (10.93,3.29)   ;
%Shape: Rectangle [id:dp9869693648575217] 
\draw  [dash pattern={on 0.75pt off 0.75pt}] (135.01,241.71) -- (265,241.71) -- (265,299.04) -- (135.01,299.04) -- cycle ;
%Straight Lines [id:da8430327688830221] 
\draw    (426.71,288) -- (464.71,288) ;
\draw [shift={(466.71,288)}, rotate = 180] [color={rgb, 255:red, 0; green, 0; blue, 0 }  ][line width=0.75]    (10.93,-3.29) .. controls (6.95,-1.4) and (3.31,-0.3) .. (0,0) .. controls (3.31,0.3) and (6.95,1.4) .. (10.93,3.29)   ;
%Rounded Rect [id:dp6247114974480956] 
\draw   (229,163.92) .. controls (229,160.19) and (232.02,157.17) .. (235.76,157.17) -- (256.02,157.17) .. controls (259.75,157.17) and (262.78,160.19) .. (262.78,163.92) -- (262.78,184.79) .. controls (262.78,188.52) and (259.75,191.55) .. (256.02,191.55) -- (235.76,191.55) .. controls (232.02,191.55) and (229,188.52) .. (229,184.79) -- cycle ;
%Straight Lines [id:da47921870920391507] 
\draw    (244.96,157.44) -- (244.96,142.2) ;
\draw [shift={(244.96,140.2)}, rotate = 90] [color={rgb, 255:red, 0; green, 0; blue, 0 }  ][line width=0.75]    (10.93,-3.29) .. controls (6.95,-1.4) and (3.31,-0.3) .. (0,0) .. controls (3.31,0.3) and (6.95,1.4) .. (10.93,3.29)   ;
\draw   (235.33,130.57) .. controls (235.33,125.17) and (239.71,120.8) .. (245.1,120.8) .. controls (250.49,120.8) and (254.87,125.17) .. (254.87,130.57) .. controls (254.87,135.96) and (250.49,140.33) .. (245.1,140.33) .. controls (239.71,140.33) and (235.33,135.96) .. (235.33,130.57) -- cycle ; \draw   (235.33,130.57) -- (254.87,130.57) ; \draw   (245.1,120.8) -- (245.1,140.33) ;

% Text Node
\draw (277.8,117.87) node    {$\underline{\mathbf{X}}$};
% Text Node
\draw (307.56,259.2) node    {$M,\mathbf{X}$};
% Text Node
\draw (441.44,248.1) node   [align=left] {$\displaystyle \hat{M}$};
% Text Node
\draw (341.33,116.53) node    {$\mathbf{Z}$};
% Text Node
\draw (116.88,116.53) node    {$M$};
% Text Node
\draw (201.56,128.52) node    {$\Gamma _{S}$};
% Text Node
\draw (392.18,268.71) node    {$\beta _{R}$};
% Text Node
\draw (202.22,174.69) node    {$S$};
% Text Node
\draw (311.85,130.92) node    {$W^{n}$};
% Text Node
\draw (155.22,174.69) node    {$\overline{M}$};
% Text Node
\draw (311.56,163.2) node    {$S,\overline{M} ,\mathbf{x}_{\mathcal{C} '}$};
% Text Node
\draw (203.56,199.87) node    {$\alpha _{C}$};
% Text Node
\draw (204.89,284.53) node    {$\alpha _{R}$};
% Text Node
\draw (415.56,191.87) node    {$\beta _{C}$};
% Text Node
\draw (477.44,278.1) node   [align=left] {``acc'' or ``rej''};
% Text Node
\draw (246.89,175.36) node    {$\mathbf{x}_{\mathcal{C} '}$};

\end{tikzpicture}
    \caption{An illustration of Protocol~\ref{pt:com}. The dotted boxes $\alpha_C, \beta_C$ are protocol in the commit phase. The dotted boxes $\alpha_R, \beta_R$ are protocol in the reveal phase.}
    \label{fig:protocol}
\end{figure}

\begin{theorem}[Main result] \label{thm:main}
    Consider the $(n,\theta)$-UsNC model with parameter list 
    $\theta = (p,\varepsilon_A ,l_A,\varepsilon_B,l_B)$. Let $\cC$ be a 
    linear code with Hamming distance $d=2 \sigma n$. The protocol 
    $\Phi[\varepsilon,\cC,F]$ satisfies the following bounds on its completeness,
    hiding, and binding parameters:
    \begin{align}
          \delta_{c}(n,\theta,\Phi[\varepsilon,\cC,F]) &\le  8 \cdot 2^{- n \varepsilon^2},\label{EV-C} \\   
          \delta_{h}(n,\theta,\Phi[\varepsilon,\cC,F]) &\le 2 \cdot  2^{\frac{1}{2}(n + \log |\cM| - \log |\cC| - l_B)} + 8\varepsilon_B,  \label{eq:hiding} \\ 
          \delta_{b}(n,\theta,\Phi[\varepsilon,\cC,F]) &\le   (\sqrt{2}\varepsilon n +1)^2 2^{ n \left[ 
        (1-2\sigma) h\left( \frac{p-\sigma + \varepsilon}{1-2\sigma}\right) + 2\sigma \right] - l_A}  + \varepsilon_A. \label{eq:binding}
    \end{align} 
\end{theorem}
The completeness part of Theorem~\ref{thm:main} follows directly from the property of typical sets.
\begin{align}
\delta_c(n,\theta,\Phi[\varepsilon,\cC,F]) &=
\max_{m\in \ca{M}} 1- \Pr[\hat{M}= m \land F = \textrm{``acc''}] \notag\\
&=  \max_{m\in \ca{M}} 1 - \Pr[\bX \in \cC \land \bZ \in \cT_{W,\varepsilon}^n(\bX) \land f_S(\bX) = m + \oM] \notag \\
&\le \max_{m\in \ca{M}, s \in \cS} 
\max_{\bx \in f_s^{-1}(m)} 
1- \Pr_{\bZ \sim W^n(\bx)}
[\bZ \in \cT_{W,\varepsilon}^n(\bx)]  \notag \\
& \le  8 \cdot 2^{-n \varepsilon^2}, \label{eq:SCTS}
\end{align}
where inequality \eqref{eq:SCTS} follows from Lemma \ref{lem:unit}. In Section~\ref{sec:proofBC}, we show the proofs of hiding and binding parts. 

\subsection{Proof of Theorem \ref{thm:main}} \label{sec:proofBC}
\subsubsection{Proof of hiding}
To bound the hiding parameter, we interpret the protocol as a privacy amplification scheme. The goal is to show that Bob's final state is nearly independent of the committed message $m$.
The combination of map $\Gamma_S$ and adding $\bx_{\cC'}$ in the commit phase
is effectively applying a new $2$-universal hash function $f_{S,\cC'}$ (defined in equation \eqref{eq:DFN}) on the set $\cC'$ as a randomness extractor. Then, the smoothed version of leftover hash lemma ensures the message $m$ is uniform and independent on Bob's side information.
The proof is broken down into three steps.

\textbf{Step 1:} Preparation.

Define the function $f_{s,\cC'}:\cC' \to \cM$ and its ``inverse'' random map $\Gamma_{s,\cC'}: \cM \to \cC'$ as follows:
\begin{align}
\forall c' \in \cC', f_{s,\cC'}(c') &= f_{S}(c' - \bx_{\cC'}),  \label{eq:DFN}\\
\forall c' \in \cC', m \in \cM, \Gamma_{s,\cC'}(c'|m) &= \frac{|\cM|}{|\cC|} \mathbbm{1}[f_{s,\cC'}(c')=m].
\end{align}
Note that $\{f_{s,\cC'}\}_{s\in \cS}$ is also a 2-universal hash function family.

Define the concatenated map $\phi_{s,\cC'}: \cM \to \cD(\cH_B)$ as $\phi_{s,\cC'} \coloneqq W^B \circ \Gamma_{s,\cC'}$. $\bar{\phi}_{s,\cC'} = \sum_{m \in \cM}\phi_{s,\cC'}(m)/|\cM|$ is the marginal output state when the input is uniform distributed. 

\textbf{Step 2:} Apply the leftover hash lemma.

When Bob is dishonest, he receives quantum register $B$ (through channel $W^B$) and classical variables $S,\bar{M},\cC'$. With the notations in Step 1, the state of Bob conditioned on input $m$ can be expressed as
\begin{align} \label{eq:hdrho}
    \rho_{B\cC' S \bar{M}|m} = \frac{1}{|\cM| } \cdot \frac{1}{|\cS|}  \cdot \frac{|\cC|}{2^n} \sum_{s \in \cS, \bar{m} \in \cM, \cC' \in \cX^n/\cC} \proj{s} \otimes \proj{\bar{m}} \otimes \proj{\cC'} \otimes \phi_{s,\cC'}(m+\bar{m}).
\end{align}
Following the definition of hiding, we have for any $m,m'\in \cM$:
\begin{align}
    &\GTD(\rho_{B\cC'S\bar{M}|m},\rho_{B\cC'S\bar{M}|m'}) \notag \\
    ={}& \frac{1}{2} \left\|\rho_{B\cC'S\bar{M}|m} - \rho_{B\cC'S\bar{M}|m'} \right\|_1  \notag \\
    ={}& \bbE_{\cC'} \bbE_S  \sum_{m \in \cM} \frac{1}{|\cM|}
    \|\phi_{S,\cC'}(m+\om) - \phi_{S,\cC'}(m'+\om) \|_1   \notag\\
    \le{}&  \bbE_{\cC'} \bbE_S \frac{1}{2|\cM|}\sum_{\om \in \cM}  
    \left[ \|\phi_{S,\cC'} (m+\om) - \bar{\phi}_{S,\cC'}  \|_1 + \| \phi_{S,\cC'}(m'+\om) - \bar{\phi}_{S,\cC'} \|_1 \right] \label{eq:TRI}\\
    ={}&  \bbE_{\cC'} \bbE_S  \frac{1}{|\cM|}\sum_{\um \in \cM} 
    \|\phi_{S,\cC'}(\um) - \bar{\phi}_{S,\cC'}  \|_1  \label{eq:GA} \\
    {}=& \bbE_{\cC'} \bbE_S \|\phi_{S,\cC'} \times U_{\cM} - \bar{\phi}_{S,\cC'} \otimes U_{\cM}  \|_1 \notag\\
    ={}& \bbE_{\cC'} \bbE_S \|f_{S,\cC'} (W^B \times U_{\cC'}) - \bar{\phi}_{S,\cC'} \otimes U_{\cM}  \|_1, \label{eq:stp1} 
\end{align}
where \eqref{eq:TRI} follows from triangle inequality and \eqref{eq:GA} follows from group rearrangement.

Applying the smoothed version~\cite[Corollary 5.6.1]{renner2008security} of the leftover hash lemma \cite{bennett1995generalized,hastad1999pseudorandom} to equation \eqref{eq:stp1}, we have, for any function $\varepsilon_B'(\cC') \ge 0$,
\begin{align}
    \bbE_{\cC'} \bbE_S \|f_{S,\cC'} (W^B \times U_{\cC'}) - \bar{\phi}_{S,\cC} \otimes U_{\cM}  \|_1 \le \bbE_{\cC'} 2 \cdot  2^{\frac{1}{2}( \log |\cM| -H_{\min}^{\varepsilon_B'(\cC')}(\bX|B)_{W^B \times U_{\cC'}})} + 4\varepsilon_B'(\cC') \label{eq:stp2}.
\end{align}

\textbf{Step 3:} Upper-bound equation \eqref{eq:stp2} with channel parameter $l_B$ and $\varepsilon_B$.

Note that, for a set of subnormalized states $\{\tau_{\bX B|\cC'}\}_{\cC' \in \cX^n/\cC}$,  inequality \eqref{eq:stp2} can be re-written as 
\begin{align}
    \bbE_{\cC'} \bbE_S \|f_{S,\cC'} (W^B \times U_{\cC'}) - \bar{\phi}_{S,\cC} \otimes U_{\cM}  \|_1 \le \bbE_{\cC'} 2 \cdot  2^{\frac{1}{2}( \log |\cM| -H_{\min}(\bX|B)_{\tau_{\bX B|\cC'}})} + 4\GTD(\tau_{\bX B|\cC'} , W^B \times U_{\cC'}) \label{eq:JW1}
\end{align}
The key idea is to find a set $\{\tau_{\bX B|\cC'}\}_{\cC' \in \cX^n/\cC}$ with proper conditional min-entropy $H_{\min}(\bX|B)_{\tau_{\bX B|{\cC}' }}$ and short average distance $\GTD(\tau_{\bX B|\cC'} , W^B \times U_{\cC'})$.

We choose $d'$ such that $2^{d'}=|\cX^n|/|\cC|=2^n/|\cC|$. 
Note that 
\begin{align}
H^{\varepsilon_B}_{\min}(\bX|B)_{W^B \times U_{\cX^n}}\notag \notag 
=
\max_{\tau_{\bX B} \in \cD_{\bullet}(\cH_{\bX B})}
\{ 
H_{\min}(\bX|B)_{\tau_{\bX B}}|
\GTD(\tau_{\bX B, W^B \times U_{\cX^n}})
\le \varepsilon_B
\} ,\label{NH2}
\end{align}
Assume $\tau^*_{\bX B}$ is the subnormalized state achieving maximum in equation \eqref{NH2}.
Denote the marginal distribution of $\tau^*_{\bX B}$ by $P_{X^x}$. 
That is, $\tau^*_{\bX B}$ is written as 
\begin{align}
    \tau^*_{\bX B} =\sum_{\bx \in \cX^{n}} P_{X^x}(\bx)
\tau^*_{ B|\bx} \otimes |\bx\rangle \langle \bx|.
\end{align}
We construct the states $\tau_{\bX B|{\cC}'}$ and $\tilde{\tau}_{\bX B|{\cC}'}$ as
\begin{align}
\tau_{\bX B|{\cC}'}
&:=\sum_{\bx \in {\cC}'} \min(2^{d'},\frac{1}{P_{X^x}({\cC}')})
P_{X^x}(\bx)
\tau^*_{B|\bx} \otimes |\bx\rangle \langle \bx| \\
\tilde{\tau}_{\bX B|{\cC}'}
&:=\sum_{\bx \in {\cC}'} 2^{d'} P_{X^x}(\bx)
\tau^*_{B|\bx} \otimes |\bx\rangle \langle \bx|
\end{align}
Since $\tau_{\bX B|{\cC}'}\le \tilde{\tau}_{\bX B|{\cC}'} \le 2^{d'} \tau^*_{\bX B}$, we have
\begin{align}
H_{\min}(\bX|B)_{\tau_{\bX B|{\cC}' }} &\ge 
H_{\min}(\bX|B)_{\tilde{\tau}_{\bX B|{\cC}' }} \notag\\
&\ge  H_{\min}(\bX|B)_{\tau^*_{\bX B}} - d' \notag \\
&= H^{\varepsilon_B}_{\min}(\bX|B)_{W^B \times U_{\cX^n}}-d'.\label{eq:JJ1}
\end{align}

We also have
\begin{align}
& \bbE_{\cC'} \GTD(\tau_{\bX B|\cC'}, W^B\times U_{\cC'}) \notag\\
={}&\mathbb{E}_{{\cC}'}    
\frac{1}{2}\|
\tau_{\bX B|{\cC}'}-W^B\times U_{{\cC}'} \|_1
+\frac{1}{2}|
\tr \tau_{\bX B|{\cC}'}-\tr W^B\times U_{{\cC}'} | \notag \\
={}&\mathbb{E}_{{\cC}'}    
\frac{1}{2}\|
\tau_{\bX B|{\cC}'}-W^B\times U_{{\cC}'} \|_1
+ \frac{1}{2}|\min(2^{d'}P_{X^x}({\cC}'),1) -1| \notag
\\
={}&\mathbb{E}_{{\cC}'}\frac{1}{2}\sum_{\bx \in {\cC}'}
\|\min(2^{d'},\frac{1}{P_{X^x}({\cC}')})
P_{X^x}(\bx)\tau_{\bX B|\bx}-2^{-d} W^B(\bx) \|_1
+ \frac{1}{2}|\min(2^{d'}P_{X^x}({\cC}'),1) -1| \notag
\\
\le {}&\mathbb{E}_{{\cC}'}\frac{1}{2}\sum_{\bx \in {\cC}'}\Big(
\|2^{d'} P_{X^x}(\bx)\tau_{\bX B|\bx}-2^{-d} W^B(\bx) \|_1
+
|2^{d'}-\min(2^{d'},\frac{1}{P_{X^x}({\cC}')})| P_{X^x}(\bx)\Big)
+ \frac{1}{2}|\min(2^{d'}P_{X^x}({\cC}'),1) -1| \notag
\\
\le {}&\mathbb{E}_{{\cC}'}\frac{1}{2}\sum_{\bx \in {\cC}'}
\Big(
\|2^{d'} P_{X^x}(\bx)\tau_{\bX B|\bx}-2^{-d} W^B(\bx) \|_1\Big)
+
|2^{d'}P_{X^x}({\cC}')-\min(2^{d'}P_{X^x}({\cC}'),1)| 
+ \frac{1}{2}|\min(2^{d'}P_{X^x}({\cC}'),1) -1| \notag
\\
= {}&\mathbb{E}_{{\cC}'}\frac{1}{2}\sum_{\bx \in {\cC}'}
\Big(
\|2^{d'} P_{X^x}(\bx)\tau_{\bX B|\bx}-2^{-d} W^B(\bx) \|_1\Big)
+
\frac{1}{2} |2^{d'}P_{X^x}({\cC}')-1|  \notag
\\
\stackrel{(a)}{\le}  {}&\mathbb{E}_{{\cC}'}\sum_{\bx \in {\cC}'}
\Big(
\|2^{d'} P_{X^x}(\bx)\tau_{\bX B|\bx}-2^{-d} W^B(\bx) \|_1\Big) \notag
\\
={}&\|
\tau_{\bX B}-W^B\times U_{\cX^n} \|_1
\le \|
\tau_{\bX B}-W^B\times U_{\cX^n} \|_1
+|\tr \tau_{\bX B}-\tr W^B\times U_{\cX^n} |
\le 2 \varepsilon_B, \label{eq:JJ2}
\end{align}
where step (a) follows from
\begin{align}
\sum_{\bx \in {\cC}'}
\Big(
\|2^{d'} P_{X^x}(\bx)\tau_{\bX B|\bx}-2^{-d} W^B(\bx) \|_1\Big)
\ge %\le
\|\sum_{\bx \in {\cC}'}
\Big(
2^{d'} P_{X^x}(\bx)\tau_{\bX B|\bx}-2^{-d} W^B(\bx) \Big)\|_1
=|2^{d'}P_{X^x}({\cC}')-1| .
\end{align}

Substitute inequalities \eqref{eq:JJ1} and \eqref{eq:JJ2} into inequality \eqref{eq:JW1}, we have 
\begin{align*}
    \GTD(\rho_{B\cC'S\bar{M}|m},\rho_{B\cC'S\bar{M}|m'}) \le 2^{\frac{1}{2}\left( n + \log |\cM| - \log |\cC| - l_B\right)} + 8 \varepsilon_B,
\end{align*}
which implies
\begin{align*}
    \delta_{h}(n,\theta,\Phi[\varepsilon,\cC,F]) \le 2^{\frac{1}{2}\left( n + \log |\cM| - \log |\cC| - l_B\right)} + 8 \varepsilon_B.
\end{align*}

\subsubsection{Proof of binding}
As shown in Figure \ref{fig:bc}(c) The core of the binding proof is to show that it is impossible for a cheating Alice to open her commitment with two different messages.
A successful attack requires her to find two distinct codewords, $\bx_0$ and $\bx_1$, such that  they are accepted with a single channel output $\bZ$.
The bound on binding parameter relies on two facts. First, for any two distinct codewords, the intersection of their corresponding typical sets, $\cT_{W,\varepsilon}^n(\bx_0) \cap \cT_{W,\varepsilon}^n(\bx_1)$, is small. 
Second, Alice's ability to concentrate the probability distribution of her channel output $\bZ$ is limited by the entropic constraint. Combining the small target space and the diffused output distribution makes the attack almost impossible.

Specifically, a general commit strategy $\cA_C$ generates random variables $(S, \uX, \oM, \cC', R)$ subjected to some joint distribution.~\footnote{Alice may generate an additional quantum system correlated to these classical variables. However, because Alice and Bob do not have any quantum interaction, we can always trace out this quantum system to get the current classical strategy.} $R$ is the extra information flowing from $\cA_C$ to $\cA_{R0}$ and $\cA_{R1}$.
$\uX$ is the input of the channel $W^A$.
$S$, $\oM$ and $\cC'$ are sent to Bob directly.
Similarly, a general reveal strategy is a probabilistic map $(S,\uX, \oM, \cC', R) \mapsto (\bX, M)$. Specifically, we have 
\begin{align*}
    \cA_{R0}: & (S,\uX,\oM, \cC', R) \mapsto (\bX_0, M_0) ,\\
    \cA_{R1}: & (S,\uX,\oM, \cC', R) \mapsto (\bX_1, M_1).
\end{align*}
The overall probability distribution of $(S, \uX, \oM, \cC', \bX_0, \bX_1, M_0, M_1)$ after applying $\cA_C, \cA_{R0}$ and $\cA_{R1}$ is denoted by $P_{\square}$.
Denote by $\beta$ the test function that determines the acceptance in the reveal phase, then 
\begin{align}
    & \delta_{b}(n,\theta,\Phi[\varepsilon,\cC,F]) \notag \\
    ={} & \max_{P_{\square} }
    \Pr_{\substack{
    (S,\bZ,\oM,\cC', \bX_0,\bX_1,M_0,M_1) \sim W^A \circ P_{\square}
    }}  
    \left[
    \begin{array}{l}
    M_0 \ne M_1 \; \land  \\
    F_1 = F_2= \textrm{``acc''}
    \end{array}
    \right]  \notag \\ %%%%%%%%%%%%%%%%%%%%%%%%%%%---line 2---
    \le {}&
    \max_{P_{S,\uX,\oM,\cC',\bX_0,\bX_1} } \max_{m_0\ne m_1}
    \Pr_{ \substack{
    (S,\bZ,\oM,\cC',\bX_0,\bX_1) \sim W^A \circ P_{S,\uX,\oM,\cC',\bX_0,\bX_1}
    }}  
    \left[
    \begin{array}{l}
    \beta(S,\bZ,\oM,\cC',\bX_0,m_0) = \textrm{``acc'' } \; \land \\
    \beta(S,\bZ,\oM,\cC',\bX_1,m_1) = \textrm{``acc''}
    \end{array}
    \right]  \notag \\%%%%%%%%%%%%%%%%%%%%%%%%---line 3---
    ={} &   \max_{ \substack{ P_{S,\uX,\oM,\cC',\bX_0,\bX_1} \\
    m_0 \ne m_1 } } \Pr_{\dots}
    \left[
    \begin{array}{l}
    \bX_0 \in \cC \land f_S(\bX_0) = m_0 + \oM \land \bZ \in \cT_{W,\varepsilon}^n(\bX_0 + \bx_{\cC'}) \;\land \\
    \bX_1 \in \cC \land f_S(\bX_1) = m_1 + \oM \land \bZ \in \cT_{W,\varepsilon}^n(\bX_1 + \bx_{\cC'})
    \end{array}
    \right] \notag  \\%%%%%%%%%%%%%%%%%%%%%%---line 4---
     \le {} & \max_{  P_{S,\uX,\oM,\cC',\bX_0,\bX_1}} 
    \Pr_{\dots} \left[ \bX_0, \bX_1 \in \cC \land \bX_0 \ne \bX_1 \land 
    \bZ - \bx_{\cC'} \in \cT_{W,\varepsilon}^n(\bX_0) \cap \cT_{W,\varepsilon}^n(\bX_1) 
    \right]  \label{eq:INS} \\%%%%%%%%%%%%%%%%%%%%%%---line 5---
    = {} & \max_{P_{S,\uX,\oM,\cC',\bX_0,\bX_1}} 
    \sum_{ \substack{\bx_0 , \bx_1 \in \cC \\
    \bx_0 \ne \bx_1 } } 
    P_{\bX_0,\bX_1} (\bx_0, \bx_1)
    \Pr_{ \substack{\bZ \sim W^A \circ P_{\uX|(\bX_0,\bX_1)= (\bx_0,\bx_1)} \\
     }} 
     \left[ 
    \bZ \in \cT_{W,\varepsilon}^n(\bx_0) \cap \cT_{W,\varepsilon}^n(\bx_1)
    \right] \notag \\%%%%%%%%%%%%%%%%%%%%%---line 6---
    \le {}& \max_{ P_{\uX,\cC',\bX_0,\bX_1}} \max_{ \substack{\bx_0 , \bx_1 \in \cC \\
    \bx_0 \ne \bx_1 }} \quad
    \Pr_{\bZ \sim W^A \circ P_{\uX|\bx_0,\bx_1}} \left[ 
    \bZ -\bx_{\cC'}\in \cT_{W,\varepsilon}^n(\bx_0) \cap \cT_{W,\varepsilon}^n(\bx_1)
    \right] \notag \\
    \le {}&  \max_{ \substack{\bx_0 , \bx_1 \in \cC \\
    \bx_0 \ne \bx_1 \\
    \cC' \in \cX^n/\cC}} \max_{P_{\uX}} \Pr_{\bZ \sim W^A \circ P_{\uX}} [\bZ -\bx_{\cC'} \in \cT_{W,\varepsilon}^n(\bx_0) \cap \cT_{W,\varepsilon}^n(\bx_1)], \label{eq:EHZ} 
\end{align}
where inequality \eqref{eq:INS} follows from equation \eqref{eq:SYM}.

Because $\forall \bx \in \cX^n,  H_{\min}^{\varepsilon_A}(\bZ)_{W^A(\bx)} \ge l_A$, there exists a $Q_{\bZ|\bx} \in \Delta_{\bullet}(\cZ^n)$ for each $\bx$ such that $H_{\min}(\bZ)_{Q_{\bZ|\bx}} \ge l_A$ and $\GTD(Q_{\bZ|\bx}, W^A(\bx)) \le \varepsilon_A$. 
Define 
$$Q_{\bZ} = \sum_{\ux \in \cX^n} P_{\uX}(\ux) Q_{\bZ|\bx},$$
then we have 
\begin{align*}
    \GTD(Q_{\bZ}, W^A \circ P_{\uX}) &\le \varepsilon_A, \\
    H_{\min}(Q_{\bZ}) &\ge l_A.
\end{align*}
The probability in equation \eqref{eq:EHZ} is bounded by
\begin{align}
    \Pr_{\bZ \sim W^A \circ P_{\uX}} \left[\bZ -\bx_{\cC'} \in \cT_{W,\varepsilon}^n(\bx_0) \cap \cT_{W,\varepsilon}^n(\bx_1) \right] & \le \Pr_{\bZ \sim Q_{\bZ -\bx_{\cC'}} } \left[ \bZ \in \cT_{W,\varepsilon}^n(\bx_0) \cap \cT_{W,\varepsilon}^n (\bx_1) \right]  + \varepsilon_{A} \notag \\
    & \le |\cT_{W,\varepsilon}^n(\bx_0) \cap \cT_{W,\varepsilon}^n (\bx_1)|  \max_{\bz} Q_{\bZ -\bx_{\cC'}}(\bz) + \varepsilon_A \notag \\
    & = |\cT_{W,\varepsilon}^n(\bx_0) \cap \cT_{W,\varepsilon}^n (\bx_1)|  \max_{\bz} Q_{\bZ}(\bz) + \varepsilon_A \notag \\
    & \le 2^{-l_A} |\cT_{W,\varepsilon}^n(\bx_0) \cap \cT_{W,\varepsilon}^n (\bx_1)| + \varepsilon_A. \label{eq:ABX}
\end{align}
To estimate the quantity $|\cT_{W,\varepsilon}^n(\bx_0) \cap \cT_{W,\varepsilon}^n (\bx_1)|$, we have the following lemma (see Appendix \ref{app:pf} for the proof).
\begin{lemma} \label{lem:typ}
    Let $W$ be a BSC with transition probability $p$ ($p<\frac{1}{2}$), and $\bx,\by \in \{0,1\}^n$ be two binary sequences with Hamming distance $d_H(\bx,\by) = 2\sigma n$. For any $ 0 < \varepsilon < \frac{1}{2}-p$,
    \begin{align}
        |\cT_{W,\varepsilon}^n (\bx) \cap \cT_{W,\varepsilon}^n(\by)| \le
        \left\{
        \begin{array}{ll}
         (\sqrt{2}\varepsilon n + 1)^2 2^{ n \left[ 
        (1-2\sigma) h\left( \frac{p- \sigma + \varepsilon}{1-2\sigma}\right) + 2\sigma \right]},  &\textnormal{if $ \sigma \le p + 2\varepsilon$ }   \\
        0, & \textnormal{if $\sigma > p + 2\varepsilon$}
        \end{array}
        \right. \label{eq:LEM}
    \end{align}
\end{lemma}
    Combining inequalities \eqref{eq:ABX}, \eqref{eq:EHZ}, and Lemma \ref{lem:typ}, we obtain 
    \begin{align}
    \delta_{b}(n,\theta,\Phi[\varepsilon,\cC,F]) \le  (\sqrt{2}\varepsilon n + 1)^2 2^{ n \left[ 
    (1-2\sigma) h\left( \frac{p- \sigma + \varepsilon}{1-2\sigma}\right) + 2\sigma \right] -l_A} + \varepsilon_A.     
    \end{align}

\subsection{Asymptotic rate}
As a generalization of i.i.d. channel model, the UsNC is expected to have good asymptotic behavior under certain conditions. This intuition is expressed in Theorem~\ref{thm:asym}.

Before that, we define the function $g_p: [0, p] \to  [2p , h(p)]$ as
\begin{align} \label{eq:FGP}
g_p(\sigma) = 
(1-2\sigma) h\left( \frac{p-\sigma}{1-2\sigma}\right) + 2\sigma.
\end{align}

Apparently, $g_p$ is monotonically decreasing, thus the inverse function $g_p^{-1}$ is well-defined.

\begin{theorem} \label{thm:asym}
    If a sequence of UsNC $\{(n,\theta_n)\}$ with $\theta_n= (p,\varepsilon_{A,n} ,l_{A,n},\varepsilon_{B,n}, l_{B,n})$ satisfies
    \begin{align}
        & \lim_{n \to \infty} \varepsilon_{A,n} =
        \lim_{n \to \infty} \varepsilon_{B,n} =0 , \label{eq:CC1}\\
        & \liminf_{n \to \infty} \frac{l_{A,n}}{n} = \xi_A \in [2p, h(p)], \label{eq:CC2}\\
        & \liminf_{n \to \infty} \frac{l_{B,n}}{n} = \xi_B \in [0,h(p)], \label{eq:CC3}
    \end{align}
    then the following rate is achievable:
    \begin{align*} 
        R = \xi_B -  h(2g_p^{-1}(\xi_A)).
    \end{align*}
\end{theorem}

\begin{proof}
For any real number $\varepsilon'$, we choose a protocol sequence $\{\Phi_n[\varepsilon_n, \cC_n, F_n]\}$ with the code sequence $\{\cC_n\}$ and a sequence of balanced UHF families $\{F_n\}$, where $F_n = \{f_{s,n}\}_{s\in \cS}: \cC_n \to \cM_n$, satisfying the following conditions.
\begin{enumerate}[label = (\Roman*)]
    \item $\varepsilon_n = \frac{1}{\sqrt[3]{n}}$. \label{it:cond0}
    \item $\cC_n$ is a code with Hamming distance $2 (g_{p}^{-1}(\xi_A) + \varepsilon') n$. \label{it:cond1}
    \item $\log | \ca{C}_n|=  (1-h(2 g_{p}^{-1}(\xi_A) +  2\varepsilon' )-\varepsilon')n$. \label{it:cond2}
    \item $\log | \ca{M}_n|=  (\xi_B - h(2g_p^{-1}(\xi_A) -3 \varepsilon') - \varepsilon')n$. \label{it:cond3}
\end{enumerate}

Condition \ref{it:cond3} implies that, for any $ R' < R$, there exists an $\varepsilon'>0$ such that $\log |\cM_n|/n \to R'$.

Next we show 
$\lim_{n\to \infty}\delta_{c}(\theta, \Phi_n[\varepsilon_{n},\cC_n,F_n]) =0
$.
The Gilbert-Varshamov bound guarantees the existence of linear codes $\{\ca{C}_n\}$ to satisfy \ref{it:cond1} and \ref{it:cond2} for large $n$.
Combining the completeness result in equation \eqref{EV-C} and the condition~\ref{it:cond0} we have 
\begin{align*}
    \lim_{n\to \infty}\delta_{c}(n,\theta, \Phi_n[\varepsilon_{n},\cC_n,F_n])=\lim_{n \to \infty} 8 \cdot 2^{-4n\varepsilon_n^2} =0.
\end{align*}

To show
    $\lim_{n\to \infty}\delta_{h}(\theta, \Phi[\varepsilon_{n},\cC_n,F_n])=0$, 
we recall the hiding result in equation~\eqref{eq:hiding}.
With equation~\eqref{eq:CC1}, it is sufficient to show that the part $2^{\frac{1}{2}(\log |\cM| + n - \log |\cC| - l_B)}$ 
in equation~\eqref{eq:hiding} goes to zero.
It holds because the exponent term goes to $-\infty$:
\begin{align*}
&\log |\cM_n| + n - \log |\cC_n| - l_{B} \\
={}& n \left[ (\xi_B - h(2g_p^{-1}(\xi_A))-\varepsilon')+ 1
-(1-h(2g_p^{-1}(\xi_A) - 2\varepsilon' )-\varepsilon') -
\frac{l_B}{n} \right] \\
={}& n\left[ \xi_B - \frac{l_B}{n} +  h(2g_p^{-1}(\xi_A) - 2\varepsilon' ) - h(2g_p^{-1}(\xi_A) - 3 \varepsilon') \right] \to - \infty .
\end{align*}

To show 
    $\lim_{n\to \infty}\delta_{b}(\theta, \Phi[\varepsilon_{n},\cC_n,F_n])=0$, 
we recall the binding result in equation~\eqref{eq:binding}.
With equation~\eqref{eq:CC1}, it is sufficient to show that the part $
2^{ n \left[ 
        (1-2\sigma) h\left( \frac{p-\sigma + 2\varepsilon}{1-2\sigma}\right) + 2\sigma \right]-l_A}
$ goes to zero.
From condition \ref{it:cond1}, we have $\sigma=g_p^{-1}(\xi_A) + \varepsilon' $, i.e., $\xi_A = g_p(\sigma -\varepsilon')$. 
Then 
\begin{align*}
& \limsup_{n \to \infty} \; (1-2\sigma) h\left( \frac{p-\sigma + 2\varepsilon_n}{1-2\sigma}\right) + 2\sigma
-\frac{l_A}{n} \\
={}&
(1-2\sigma) h\left( \frac{p-\sigma}{1-2\sigma}\right) + 2\sigma
-g_p(\sigma - \varepsilon') \\
={}& g_p(\sigma)-g_p(\sigma - \varepsilon')<0.
\end{align*}
Therefore we have $
2^{ n \left[ 
        (1-2\sigma) h\left( \frac{p-\sigma + 2\varepsilon}{1-2\sigma}\right) + 2\sigma \right]-l_A} \to 0
$.
\end{proof}

Figure~\ref{fig:rate} is a numeric illustration of the asymptotic rate versus channel parameters when the channel parameter sequence $\{\theta_n\}$ satisfies \eqref{eq:CC1}, \eqref{eq:CC2}, \eqref{eq:CC3} and $p=0.1$. Commitment becomes impossible when $\xi_A\lesssim 0.82 h(p)$ or $\xi_B = 0$. Additionally, the rate bears higher tolerance on channel $W^B$ than $W^A$.

\begin{figure}[h!]
    \centering
    \includegraphics[width=0.6\linewidth]{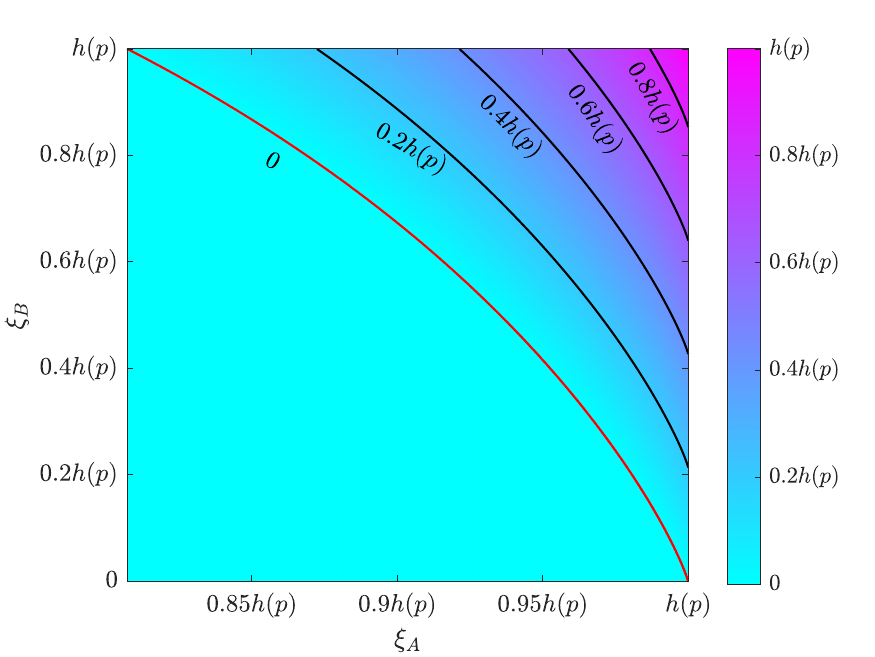}
    \caption{The rate $R$ as a function of $\xi_A$ and $\xi_B$ when $p=0.1$ ($h(p) \approx 0.469$). The X and Y axes are $\xi_A$ and $\xi_B$ respectively. The rate hits zero at the red line and below, and hits maximum value $h(p)$ at the point $\xi_A=\xi_B=h(p)$.}
    \label{fig:rate}
\end{figure}

%Some discussion on the i.i.d. case
Theorem \ref{thm:asym} provides a commitment rate for general channels $W^A$ and $W^B$.
In the special case when $W^A$ and $W^B$ are fixed as i.i.d. BSC, we demonstrate that Theorem~\ref{thm:asym} recovers the earlier result of~\cite{winter2003commitment,crepeau2020commitment}. 
To establish this connection, we first introduce Lemma~\ref{lem:iidc}, which examines the asymptotic behavior of i.i.d. BSC (see Appendix~\ref{app:plem} for the proof).

\begin{lemma} \label{lem:iidc}
    For $n$-fold BSC $W^n: \cX^n \to \cZ^n$ ($\cX=\cZ =\{0,1\}$) with transition probability $p$, there exist parameters $\mu_n, l(n,p)$ such that 
    \begin{align*}
    H_{\min}^{\mu_n} (\bZ)_{W^n(\bx)} & \ge l(n,p) \textnormal{ for all } \bx \in \{0,1\}^n, \\
    H_{\min}^{\mu_B}(\bX|B)_{W^B \times U_{\cX^n}}  & \ge l(n,p), 
    \end{align*}
    and
    \begin{align*}
        & \lim_{n \to \infty} \mu_n = 0, \\
        & \lim_{n \to \infty} \frac{l(n,p)}{n} = h(p).
    \end{align*}
\end{lemma}
Lemma \ref{lem:iidc} effectively defines a sequence of UsNC from fixed BSC. Then we obtain the following corollary.
\begin{corollary} \label{coro:iidc}
If $W^n, W^A$ and $W^B$ are $n$-fold BSC with transition probability $p, p_A$ and $p_B$ respectively, then the following commitment rate is achievable:
\begin{align} \label{eq:iidrate}
    R = h(p_B) - h(2g_p^{-1}(h(p_A))). 
\end{align}
\end{corollary}

\begin{proof}
    According to Lemma~\ref{lem:iidc}, the $n$-fold BSC can be seen as UsNC model with parameter $\theta_n =(p,\varepsilon_{A,n},l_{A,n},\varepsilon_{B,n}, l_{B,n})$ can be taken as  $\varepsilon_{A,n} = \varepsilon_{B,n} = \mu_n, l_{A,n} = l(n,p_A), l_{B,n} = l(n,p_B)$. The rate in equation~\eqref{eq:iidrate} follows immediately from Theorem~\ref{thm:asym}.
\end{proof}

Cr\'{e}peau et al.~\cite{crepeau2020commitment} consider the case $p_A=p_B \le p$, as in Corollary~\ref{coro:iidc}, and derive the capacity 
\begin{align} \label{eq:capa}
    C = h(p_A) - h\left( \frac{p-p_A}{1-2p_A}\right).
\end{align}
When $p_A =p_B < p$, denote $p'  = \frac{p-p_A}{1-2p_A}$, then the following argument shows that the capacity $C$ exceeds the rate $R$ presented in equation~\eqref{eq:iidrate}:
\begin{align*}
    & h(p_A)=h\left( \frac{p-p'}{1-2p'}\right) < 
    g_p(p') < g_p \left(\frac{p'}{2} \right) \\
    \implies & 2g_p^{-1}(h(p_A)) > p' \\
    \implies & h(2g_p^{-1}(h(p_A))) > h(p').
\end{align*}
This gap between our achievable rate $R$ and the capacity $C$ in equation~\eqref{eq:capa} is expected and reflects the trade-off between generality and tightness. 
The capacity $C$ is derived under a restricted i.i.d. adversary, so this additional structure can be exploited to obtain a higher rate,
while our rate $R$ is derived under a stronger class of adversaries.
When $p_A = p_B =p$, we have $g_p^{-1}(h(p_A))=0$, so both $R$ and $C$ reduce to the previously established capacity $h(p)$ in~\cite{winter2003commitment,imai2004rates}, indicating that our general model correctly recovers known results.

\section{Application in Noisy-Quantum-Storage Model} \label{sec:Application}
To demonstrate the utility of our UsNC model, we present an example of such channel model constructed in the noisy-quantum-storage (NQS) model~\cite{damgaard2007tight,wehner2008cryptography,wehner2008cryptographyb,schaffnerchristian2009robust}.
As a natural generalization of bounded-quantum-storage model, the NQS model is particularly relevant because building noiseless quantum memory is still challenging with current technology.
More specifically, 
the noisy quantum storage is characterized by a quantum channel $\cF: \cB(\cH_{\text{in}}) \to \cB(\cH_{\text{out}})$, which can be forced by certain time delay $T$ in the protocol. 
Any quantum information stored by the adversary must go through channel $\cF$, which inherently weakens the adversary's capability.
In addition, a noiseless quantum channel is assumed to be available in the NQS model.

\subsection{Protocol}
The following protocol constructs UsNC in the NQS model characterized by channel $\cF$.

\begin{algorithm}[h!]
\caption{A noisy classical channel from noisy quantum storage} \label{pt:chnqs}
\textbf{Channel input}: $\bX \in \{0,1\}^n$ \\
\textbf{Channel output}: $\bZ \in \{0,1\}^n$
\begin{enumerate}[leftmargin=*]
    \item Alice generates a random bit string $\Theta^n \in \{0,1\}^n$. For each $i \in [n]$, Alice sends to Bob a quantum register $B_i$ in state $\ket{\psi}_{B_i} = H^{\Theta_i} \ket{X_i}$,
    where $H$ is Hadamard operator.
    \item Bob generates a random bit string $\Theta'^n \in \{0,1\}^n$. For each $i \in [n]$, Bob measures observable $(\sigma_z + (-1)^{\Theta'_i} \sigma_x)/\sqrt{2}$ on register $B_i$ and obtains measurement result $K^n \in \{0,1\}^n$.
    \item Alice waits for a predetermined time $T$, and sends the string $\Theta^n$ to Bob.
    \item For each $i \in [n]$, Bob computes the final output $Z_i = K_i \oplus (\Theta_i \cdot \Theta'_i)$.
\end{enumerate}
\end{algorithm}

The properties of the channel constructed in Protocol \ref{pt:chnqs} are formally stated in the following theorem. The proof is deferred to Section \ref{sec:proofnqs}.
\begin{theorem} \label{thm:nqsc}
    For any $\lambda_A, \lambda_B \in (0,\frac{1}{2})$, in the noisy-quantum-storage model characterized by channel $\cF: \cB(\cH_{in}) \to \cB(\cH_{out})$, Protocol \ref{pt:chnqs} constructs a $(n,\theta)$-UsNC with  $\theta = (\sin^2(\pi/8), \varepsilon_A, l_A, \varepsilon_B, l_B)$ dependent on $\lambda_A,\lambda_B$, where
    \begin{subequations} \label{eq:LEP}
    \begin{align}
            l_A &= (h(\sin^2(\pi/8)) - 2\lambda_A)n, \label{eq:LAE}\\
            \varepsilon_A &= \exp \left(- \frac{\lambda_A^2 n}{32 (1 - \log \lambda_A)^2} \right), \label{eq:EPA}\\
            l_B &=  - \log P_{\mathrm{succ}}^{\cF}\left( \left(\tfrac{1}{2}-\lambda_B\right)n \right),  \label{eq:LBE}\\
            \varepsilon_B &= 2 \exp \left(- \frac{(\lambda_B/4)^2n}{32(2+ \log \frac{4}{\lambda_B} )^2} \right). \label{eq:EPB}
    \end{align}
    \end{subequations}
    Here $P_{\mathrm{succ}}^\cF(nR)$ is the maximized average probability of correctly transmitting message $x \in \{0,1\}^{nR}$ through channel $\cF$:
    \begin{align}
        P_{\mathrm{succ}}^\cF(nR) = \max_{\{\rho_x\}, \{D_x\} } \frac{1}{2^{nR}} \sum_{x \in \{0,1\}^{nR} } \tr (D_x \cF (\rho_x)), \notag
    \end{align}
    where $\{\rho_x\}$ is a list of elements in $\cB(\cH_{in})$ that defines an encoder map $\{0,1\}^{nR} \to \cB(\cH_{in})$, and $\{ D_x \}$ is a decoding positive operator-valued measurement on $\cH_{\text{out}}$.
\end{theorem}

\begin{remark}
    The free parameters $\lambda_A,\lambda_B$ characterize the trade-offs between the bound $l_A,l_B$ and the smoothing parameters $\varepsilon_A,\varepsilon_B$.
    A choice of smaller $\lambda_{A}$ ($\lambda_B$) leads to a stronger bound $l_{A}$ ($l_B$) at the expense of a more relaxed smoothing parameter $\varepsilon_{A}$ ($\varepsilon_B$).
\end{remark}

For equation~\eqref{eq:LBE}, if the channel $\cF$ takes the tensor product form $\cF = \cN^{\otimes n\nu}$ for storage rate $\nu$, then for some common channel like entanglement-breaking channel, there exists a positive function $\gamma^{\cN}(\cdot)$ such that for large $n$,
\begin{align}
    \forall R > C_{\cN}, P_{\mathrm{succ}}^\cF(nR) \le 2^{-\gamma^{\cN}(\frac{R}{\nu})n\nu}, \notag
\end{align}
where $C_{\cN}$ is the classical capacity of channel $\cN$.
We then have the bound on $l_B$
\begin{align}
    l_B \ge n \nu \gamma^{\cN}\left(\frac{1/2 - \lambda_B}{\nu} \right). \notag
\end{align}

The combination of Protocol~\ref{pt:chnqs} and Protocol~\ref{pt:com} establishes a string commitment protocol in the NQS model. We have the following corollary derived from Theorem~\ref{thm:main} and Theorem~\ref{thm:nqsc}.
\begin{corollary}
    The commitment protocol obtained by combining Protocol~\ref{pt:chnqs} and protocol $\Phi[\varepsilon,\cC,F]$ (Protocol~\ref{pt:com}) satisfies:
    \begin{align*}
        \delta_{c} &\le  8 \cdot 2^{- n \varepsilon^2}, \\  
        \delta_{h} &\le 2 \cdot  2^{\frac{1}{2}(n + \log |\cM| - \log |\cC| - l_B)} + 8\varepsilon_B,  \\ 
        \delta_{b} &\le   (\sqrt{2}\varepsilon n +1)^2 2^{ n \left[ 
        (1-2\sigma) h\left( \frac{p-\sigma + \varepsilon}{1-2\sigma}\right) + 2\sigma \right] - l_A}  + \varepsilon_A,
    \end{align*}
    where $p = \sin^2(\pi/8)$, and $l_A,\varepsilon_A,l_B,\varepsilon_B$ are given in equations \eqref{eq:LEP}.
\end{corollary}

\subsection{Achievable commitment rate}
The commitment rate in the NQS model is
the ratio of committed bits to the number of qubits transmitted via the noiseless quantum channel, given by
\begin{align*}
    R_{\mathrm{NQS}} = \frac{\log |\cM|}{n}.
\end{align*}
Since Protocol~\ref{pt:chnqs} yields one (honest) binary symmetric channel per transmitted qubit, we use the same notation $n$ for both quantities here.
Similar to Definition~\ref{def:crate}, $R_{\mathrm{NQS}}$ is an achievable commitment rate in the NQS model if, for any $R'<R_{\mathrm{NQS}}$, there exists a infinite sequence of commitment protocol such that $\delta_c,\delta_h,\delta_b$ vanishes and $\log |\cM|/n \to R'$.

To determine an explicit achievable commitment rate, we assume $\cF$ is an identity channel with fixed dimension $d$, i.e., $\text{dim} \cH_{\text{in}} = \text{dim} \cH_{\text{out}} = d$. This corresponds to the bounded noiseless quantum storage model. 
In this scenario, the quantity $P_{\mathrm{succ}}^{\cF}(nR)$ becomes
\begin{equation} \notag
\begin{aligned}
    P_{\mathrm{succ}}^{\cF}(nR) &= \max_{\{D_x\}, \{\rho_x\} }\frac{1}{2^{nR}}\sum_{x \in \{0,1\}^{nR}} \tr(D_x \rho_x) \\
    &\le \max_{\{D_x\}}\frac{1}{2^{nR}} \sum_{ x \in \{0,1\}^{nR} } \tr(D_x) \\
    &= \frac{1}{2^{nR}} \tr (\mathbb{I}_d) = 2^{-nR + d}.
\end{aligned}
\end{equation}
By setting $\lambda_A = \lambda_B = \frac{1}{\sqrt[3]{n}}$ in Theorem \ref{thm:nqsc}, we observe the following limits:
\begin{subequations} \label{eq:LEPLIM}
   \begin{align}
    \lim_{n \to \infty} \varepsilon_A &= 0,\\
    \lim_{n \to \infty} \varepsilon_B &= 0, \\
    \lim_{n \to \infty} \frac{l_A}{n} &= h(\sin^2(\pi/8)), \\
    \lim_{n \to \infty} \frac{l_B}{n} &= \frac{1}{2}.
\end{align} 
\end{subequations}

With the above limits, we get the following corollary on achievable commitment rate in the NQS model.
\begin{corollary}
    In NQS model characterized by $d$-dimensional identity channel $\cF = \text{id}_d$, the combination of Protocol \ref{pt:chnqs} and Protocol \ref{pt:com} achieves the commitment rate $R_{\mathrm{NQS}}=\frac{1}{2}$.
\end{corollary}
\begin{proof}
For the combined protocol, the rate $R_{\text{NQS}}$ equals the rate defined in equation~\eqref{eq:ratedef}. Therefore, we can apply Theorem~\ref{thm:asym}.
Substituting $p = \sin^2(\pi/8)$ and equations~\eqref{eq:LEPLIM} into Theorem~\ref{thm:asym}, we have $\xi_A = h(p), \xi_B = \frac{1}{2}$. Then the following rate is achievable:
\begin{align*}
    R_{\mathrm{NQS} }= \frac{1}{2} - h(2g_p^{-1}(h(p))).
\end{align*}
By the definition of function $g_p(\cdot)$ in equation~\eqref{eq:FGP}, 
we observe that $g_p(0) = h(p)$ and $g_p^{-1}(h(p))=0$. Therefore, the achievable rate simplifies to $R_{\mathrm{NQS}} = \frac{1}{2}$.
\end{proof}

\subsection{Proof of Theorem~\ref{thm:nqsc}} \label{sec:proofnqs}
By Definition~\ref{def:UsNC}, proving Theorem~\ref{thm:nqsc} amounts to proving the channel constructed in Protocol~\ref{pt:chnqs} satisfies conditions \textbf{(C1)}, \textbf{(C2)} and \textbf{(C3)} with the given parameter list $\theta = (\sin^2(\pi/8),l_A,\varepsilon_A, l_B, \varepsilon_B)$.
In the honest case, Alice and Bob are essentially playing a CHSH game, with the winning condition being $Z_i = X_i$. The channel mapping $X^n$ to $Z^n$ is an $n$-fold i.i.d. binary symmetric channel with error probability $p = \sin^2(\pi/8)$.
Therefore, condition \textbf{(C1)} is satisfied.
Next, we discuss the channel properties with dishonest Bob and dishonest Alice.

\subsubsection{Dishonest Bob case}

Alice's part of Protocol \ref{pt:chnqs} is identical to that of the Weak String Erasure protocol from \cite[Protocol 1]{konig2012unconditional}. Consequently, the bound on dishonest Bob, stated in Theorem~\ref{thm:BDH}, also applies in our setting. This theorem directly shows that condition \textbf{(C3)} of Definition~\ref{def:UsNC} is satisfied.

\begin{theorem}[Theorem III.3 of \cite{konig2012unconditional}, adapted] \label{thm:BDH}
    Let $\lambda_B \in (0,\frac{1}{2})$. 
    In Protocol~\ref{pt:chnqs}, for any attack by dishonest Bob using 
    noisy storage $\cF: \cB(\cH_{\text{in}}) \to \cB(\cH_{\text{out}})$, let $B$ be all the information Bob has after receiving $\Theta^n$ from Alice and $W^B:\cX^n \to \cD(\cH_{B})$ be the resulted cq channel. Then 
    \begin{align} 
        H_{\min}^{\varepsilon_B}(\bX|B)_{W^B \times U_{\cX^n}} \ge - \log P_{\mathrm{succ}}^{\cF}((\frac{1}{2}-\lambda_B)n), \label{eq:HB}
    \end{align}
    where 
    \begin{align}
        \varepsilon_B = 2 \exp \left(- \frac{(\lambda_B/4)^2n}{32(2+ \log \frac{4}{\lambda_B})^2} \right). \notag
    \end{align}
\end{theorem}

\subsubsection{Dishonest Alice case}
To show the bound in \eqref{eq:ch2}, note that for any set $\cX$ and any element $x\in \cX$, 
the general strategy of Alice is to prepare and send a cq state $\rho_{\Theta^n B^n|x}$ to Bob.
Then it is sufficient to show that for any cq state $\rho_{\Theta^n B^n}$, Bob's final output $\bZ$ satisfies $H_{\min}^{\varepsilon_A}(\bZ) \ge l_A$.
For this purpose, we break down the $n$ rounds of Bob's measurement and use state-independent entropic uncertainty bound.

In the $i$th round, honest Bob measures the system $B_i$ and compute $Z_i$ as described in step 4 of Protocol~\ref{pt:chnqs}. Bob's measurement together with the post-processing is essentially a POVM $\{E_z\}$ described as follows.
\begin{align*}
    E_0 &= \frac{1}{2} \left[ \pi_{0} \otimes \left( \pi_{0} + \pi_{+} \right) +  \pi_{1} \otimes \left( \pi_{1} + \pi_{+} \right) \right],\\
    E_1 &= \frac{1}{2} \left[ \pi_{0} \otimes \left( \pi_{1} + \pi_{-} \right) +  \pi_{1} \otimes \left( \pi_{0} + \pi_{+} \right) \right].
\end{align*}
where $\pi_x = \proj{x}$. 

Assume the result of previous $i-1$ rounds is $z^{i-1}$.
We bound the min-entropy of $Z_i$ conditioned on $Z^{i-1} =z^{i-1}$ with the entropic uncertainty relation~\cite{tomamichel2011uncertainty}, which states that, for two POVMs $\{M_z\}, \{N_x\}$,
\begin{align*}
    &H_{\min}(Z) + H_{\max}(X) \ge \log \frac{1}{c},\\
    &c = \max_{x,z} \| \sqrt{M_z} \sqrt{N_{x}} \|_{\infty}^2,
\end{align*}
where $X$ and $Z$ are the measurement result of the POVMs. Setting $\{M_z\} = \{E_0, E_1\}, \{N_x\} = \{\mathbb{I}\}$, we have 
\begin{align*}
    H_{\min} (Z_i | Z^{i-1} = z^{i-1}) \ge \log \frac{1}{c}, 
\end{align*}
where $c = \cos^2(\pi/8)$.
With the following argument, we transform the above bound on min-entropy into bound on Shannon entropy.
\begin{align}
    & H_{\min}(Z_i|Z^{i-1} = z^{i-1}) \ge \log \frac{1}{c} \notag\\
    \implies & \max_{z\in \{0,1\}} P_{Z_i|Z^{i-1}=z^{i-1}}(z) \le c \notag \\
    \implies & H(Z_i|Z^{i-1} = z^{i-1}) \ge h(c). \label{eq:SCS}
\end{align}

To bound the min-entropy of $Z^n$, we use the following theorem derived from Azuma's inequality.
\begin{theorem}[Theorem 3.1 of \cite{damgaard2007tight}] \label{thm:amz}
    Let $\bZ = Z_1 \dots Z_n$ be $n$ (not necessarily independent) random variables over alphabet $\cZ$, and let $h\ge 0$ be such that
    \begin{align*}
        H(Z_i|Z^{i-1} = z^{i-1})\ge h
    \end{align*}
    for all $i \in [n]$ and $z_1, \dots, z_n \in \cZ$. Then for any $0<\lambda<\frac{1}{2}$,
    \begin{align} \label{eq:amz}
        H_{\min}^{\varepsilon}(\bZ) \ge (h-2\lambda)n,
    \end{align}
    where $ \varepsilon= \exp \left( \frac{\lambda^2 n}{32 \log ( |\cZ|/ \lambda)^2} \right)$.
\end{theorem}

Combining Theorem~\ref{thm:amz} and inequality~\eqref{eq:SCS}, we have
\begin{align}
    \forall \lambda_A \in (0,\frac{1}{2} ),  H_{\min}^{\varepsilon_A}(\bZ) \ge (h(\sin^2(\pi/8))-2\lambda_A)n, 
\end{align}
where $\varepsilon_A = \exp \left( \frac{\lambda_A^2 n}{32 (1 - \log \lambda_A)^2} \right)$. This result shows that condition \textbf{(C2)} of Definition~\ref{def:UsNC} is satisfied.
\section{Conclusion}
In this work, we introduced the UsNC model, an adversarial channel model for two-party cryptographic tasks.
This model bridges two distinct approaches to commitment by providing a general information-theoretic structure that can be instantiated by concrete physical assumptions.
The UsNC model generalizes the i.i.d. assumption in the unfair noisy channel model by constraining the adversarial channel with global entropic bounds, thereby allowing arbitrary memory effects across channel uses.

Within the UsNC model, we constructed a string commitment protocol and proved its security. We derived the completeness error, the hiding parameter and the binding parameter of the protocol. 
Under the asymptotic condition in Theorem \ref{thm:asym}, we derived the commitment rate, which is determined by the channel noise parameter $p$ and the asymptotic ratios $l_A/n, l_B/n$.
Our commitment rate simplifies to the known commitment capacity $C=h(p)$ in earlier results~\cite{winter2003commitment,imai2004rates} when the channel $W^A$ and $W^B$ are fixed as $n$-fold BSC.
We then showed that the UsNC model can be instantiated by the physical assumption of noisy quantum storage. 
The required entropic constraints of the UsNC model are derived from physical principles in the NQS setting.

The assumption of a $n$-fold BSC for the honest parties, as specified in condition \textbf{(C1)} of Definition~\ref{def:UsNC}, is a  limitation of the UsNC model.
This assumption is not required for the protocol's completeness, because the proof relies on standard properties of typical sets (Lemma~\ref{lem:unit}) that hold for general discrete memoryless channels.
Instead, the limitation stems from the proof for the binding property.
Specifically, the bound on the binding parameter relies on Lemma~\ref{lem:typ}, whose proof requires the specific structure of BSC. 
Therefore, if Lemma~\ref{lem:typ} can be relaxed to a general discrete memoryless channel, then the UsNC model can be generalized accordingly.
This extension is left open for future study.

\section{Acknowledgement}
JW and MT were supported by the National Research Foundation, Singapore and A*STAR under its Quantum Engineering Programme (NRF2021-QEP2-01-P06).
MH was supported in part by the National Natural Science Foundation of China (Grant No. 62171212) and the General R\&D Projects of 1+1+1 CUHK-CUHK(SZ)-GDST Joint Collaboration Fund (Grant No. GRDP2025-022).

\appendices

\section{Proof of Lemma~\ref{lem:typ}} \label{app:pf}
\begin{lemma*}
      Let $W$ be a BSC with transition probability $p$ ($p<\frac{1}{2}$), and $\bx,\by \in \{0,1\}^n$ be two binary sequences with Hamming distance $d_H(\bx,\by) = 2\sigma n$. For any $ 0 < \varepsilon < \frac{1}{2}-p$,
    \begin{align}
        |\cT_{W,\varepsilon}^n (\bx) \cap \cT_{W,\varepsilon}^n(\by)| \le
        \left\{
        \begin{array}{ll}
         (\sqrt{2}\varepsilon n + 1)^2 2^{ n \left[ 
        (1-2\sigma) h\left( \frac{p- \sigma + \varepsilon}{1-2\sigma}\right) + 2\sigma \right]},  &\textnormal{if $ \sigma \le p + 2\varepsilon$ }   \\
        0, & \textnormal{if $\sigma > p + 2\varepsilon$}
        \end{array}
        \right. \label{eq:LEM2}
    \end{align}
\end{lemma*}

\begin{proof}

\noindent \textbf{Step 1}: Preparation.

Because of the translation invariance property in equation \eqref{eq:SYM}, we have
\begin{align}
   & \bz \in \cT_{W,\varepsilon}^n (\bx) \cap \cT_{W,\varepsilon}^n(\by) \notag\\
   \iff & \bz - \bx \in \cT_{W,\varepsilon}^n (\mathbf{0}) \cap \cT_{W,\varepsilon}^n(\by-\bx) \notag \\
   \iff &\left| \cT_{W,\varepsilon}^n (\bx) \cap \cT_{W,\varepsilon}^n(\by) \right| = \left| \cT_{W,\varepsilon}^n (\mathbf{0}) \cap \cT_{W,\varepsilon}^n(\by-\bx)\right| \label{eq:HDD}
\end{align}

Denote $\by'= \by -\bx$.
Partition the set $[n] \coloneqq \{1,2,\dots,n\}$ positions into two sectors
\begin{align*}
    \cI_{i} \coloneqq \{k \in [n] \mid \by' = i \}  \textnormal{ for } i = 0,1.
\end{align*}

Denote the size of the sector $|\cI_{i}|$ by $\sigma_{i} n$. Because $\HD(\bx,\by)= 2\sigma n$, we have 
\begin{align}
    \sigma_0 = 1-2\sigma, \sigma_1 = 2\sigma.
\end{align}

For any sequence $\bz \in \{0,1\}^n$ and sector $\cI \subseteq [n]$, we define the following weight function, which counts the number of $1$s
at position $\cI$.
\begin{align*}
        \wt_{\cI}(\bz) \coloneqq \sum_{x \in \cI} z_x.
\end{align*}
Then the condition in \eqref{eq:HDD} can be expressed as 
\begin{align*}
    &\frac{1}{n} \left[ \wt_{\cI_0}(\bz) + \wt_{\cI_1}(\bz) \right] \in [p-\varepsilon, p+\varepsilon], \\
    &\frac{1}{n} \left[ \wt_{\cI_0}(\bz) + \sigma_1 n - \wt_{\cI_1}(\bz) \right] \in [p-\varepsilon, p+\varepsilon]
\end{align*}

Consider a tuple $\Vec{b} = (b_0, b_1) \in \mathbb{N}^2$.
If $\Vec{b} \in [0,\sigma_{0}n] \times [0,\sigma_{1} n]$,
then $\Vec{b}$ defines a class in $\{0,1\}^n$:
\begin{align*}
    G(\Vec{b}) \coloneqq \{\bz \in \{0,1\}^n \mid \forall i \in \{0,1\}, \wt_{\cI_{i}}(\bz) = b_{i} \},
\end{align*}
Define the set
    \begin{align*}
        \Lambda \coloneqq & \left\{ \vphantom{\frac{1}{n}} \Vec{b} \in \mathbb{N}^2 \middle\vert \smash{ \frac{\Vec{b}}{n} }\in [0,\sigma_{0}] \times [0,\sigma_{1}] \right.\\
        & \frac{1}{n}(b_0 + b_1) \in \left[p -\varepsilon , p + \varepsilon \right] \land \\
        & \hspace{-0.3em} \left. \frac{1}{n}(b_0 + 2\sigma - b_1) \in \left[p - \varepsilon , p + \varepsilon \right] \right\}
    \end{align*}
    and its continuous version
    \begin{equation} \label{eq:sqr}
    \begin{aligned}
        \Lambda' \coloneqq & \left\{ \Vec{a} \in [0,\sigma_{0}] \times [0,\sigma_{1}] \mid \right.\\
        & a_0 + a_1\in \left[p -\varepsilon ,p+ \varepsilon \right] \land \\
        & \hspace{-0.3em} \left. a_0 + 2\sigma -a_1 \in \left[p - \varepsilon , p + \varepsilon \right] \right\},
    \end{aligned}        
    \end{equation}
    where $\Vec{a} = (a_0,a_1)$. Then we have 
    \begin{align*}
        \bz \in \cT_{W,\varepsilon}^n (\mathbf{0}) \cap \cT_{W,\varepsilon}^n(\by') \implies
        \bz \in \bigcup_{\Vec{b} \in \Lambda} G(\Vec{b}).
    \end{align*}

\noindent \textbf{Step 2}: Transform the estimation of $|\cT_{W, \varepsilon}^n(\bx) \cap \cT_{W,\varepsilon}^n(\by)|$ into an optimization problem.

With the above notations, the target quantity can be re-written as 
\begin{align}
    |\cT_{W,\varepsilon}^n (\bx) \cap \cT_{W,\varepsilon}^n(\by)| & \le \left| \bigcup_{\Vec{b} \in \Lambda} G(\Vec{b}) \right| \notag\\
    &\le |\Lambda| \cdot \max_{\Vec{b} \in \Lambda} |G(\Vec{b})|.  \label{eq:STPD2}
\end{align}
Note that  $\Lambda = \varnothing $ if $\sigma > p + 2\varepsilon$, thus the second case of inequality \eqref{eq:LEM2} is obtained.

If $\sigma \le p + 2\varepsilon$, we first evaluate the first term of \eqref{eq:STPD2}. The size of $\Lambda$ is essentially the number of lattice points enclosed by a square with diagonal length $2\varepsilon n$, therefore we have
\begin{align}
    |\Lambda| \le (\sqrt{2}\varepsilon n + 1)^2. \label{eq:area}
\end{align}

Then we estimate the second term of \eqref{eq:STPD2} as follows:
 \begin{align}
        \max_{\Vec{b} \in \Lambda} \left|G(\Vec{b}) \right| &= \max_{\Vec{b} \in \Lambda}  \binom{\sigma_{0}n}{b_{0}} 
        \binom{\sigma_{1}n}{b_{1}}\notag\\
        & \le \max_{\Vec{b} \in \{n\Vec{a} \mid \Vec{a} \in \Lambda'\}} \binom{\sigma_{0}n}{b_{0}} 
        \binom{\sigma_{1}n}{b_{1}} \label{eq:LNS}\\
        &= \max_{\Vec{a} \in \Lambda'} \binom{\sigma_{0}n}{a_{0}n} 
        \binom{\sigma_{1}n}{a_{1}n}\notag \\
        & \le  2^{n \max_{\Vec{a} \in \Lambda'}  \left[ \sigma_{0} h(\frac{a_{0}}{\sigma_{0}}) + \sigma_{1} h(\frac{a_{1}}{\sigma_{1}}) \right]} \label{eq:BIN}
    \end{align}
where inequality \eqref{eq:LNS} holds because $\Lambda \subseteq  \{n\Vec{a}\mid \Vec{a} \in \Lambda'\}$, and inequality \eqref{eq:BIN} follows from the property of binomial coefficient [Cite somewhere].

\noindent \textbf{Step 3}: Optimize the exponent term of \eqref{eq:BIN}.

The set $\Lambda'$ is essentially the intersection of two the squares on $\mathbb{R}^2$ plane:
\begin{align}
    \Lambda' = \Lambda'_{\mathrm{sq}} \cap [0,1-2\sigma] \times [0,2\sigma],
\end{align}
where $\Lambda'_{\mathrm{sq}}$ is the set defined by the linear constraints in equation~\eqref{eq:sqr}.

% as illustrated in figure~\ref{fig:square}.

% \begin{figure}
%     \centering
%     \input{pic/square}
%     \caption{The set $\Lambda'_{\mathrm{sq}}$ (solid-line square) is derived from the four linear constraints given in equation~\eqref{eq:sqr}. The horizontal ans vertical axes are $a_0$ and $a_1$ respectively.}
%     \label{fig:square}
% \end{figure}
% Decompose the square $\Lambda'_{\mathrm{sq}}$ into the union of horizontal lines:
% \begin{align}
%     \Lambda'_{\mathrm{sq}} = \bigcup_{a_1 \in I_1} \bigcup_{a_0 \in I_0(a_1)} (a_0,a_1), 
% \end{align}
% where $I_1 = [s-\varepsilon, s+ \varepsilon]$ and $I_0(a_1) = [p-s - \varepsilon + |a_1 - s|, p-s + \varepsilon - |a_1 - s|]$.
Then equation~\eqref{eq:BIN} is bounded as
\begin{align}
    \max_{\Vec{a} \in \Lambda'}  \left[ \sigma_{0} h(\frac{a_{0}}{\sigma_{0}}) + \sigma_{1} h(\frac{a_{1}}{\sigma_{1}}) \right] 
    &\le  \max_{\Vec{a} \in \Lambda'_{\mathrm{sq}}}  \left[ (1-2\sigma) h(\frac{a_{0}}{1-2\sigma}) + 2\sigma \right] \label{eq:NC1}\\
    &=  (1-2\sigma) h(\frac{p-\sigma + \varepsilon}{1-2\sigma}) + 2\sigma. \label{eq:NC2}
\end{align}
The equality of equation~\eqref{eq:NC2} is justified as follows.
Because $\varepsilon < \frac{1}{2}-p$, we have $(a_0,a_1) \in \Lambda'_\mathrm{sq} \implies a_0 \le p-\sigma +\varepsilon < \frac{1-2\sigma}{2}$. Then 
the function $h(\frac{a_0}{1-2\sigma})$ is monotonically increasing, and equation~\eqref{eq:NC1} reaches maximum when $a_0 = p-\sigma + \varepsilon$.
Combining equations \eqref{eq:STPD2} , \eqref{eq:area}, \eqref{eq:BIN} and \eqref{eq:NC2} yields the desired result in equation~\eqref{eq:LEM2}.
\end{proof}

\section{Proof of Lemma~\ref{lem:iidc}} \label{app:plem}
\begin{lemma*}
    For $n$-fold BSC $W^n: \cX^n \to \cZ^n$ ($\cX=\cZ =\{0,1\}$) with transition probability $p$, there exist parameters $\mu_n, l(n,p)$ such that 
    \begin{align}
    H_{\min}^{\mu_n} (\bZ)_{W^n(\bx)} & \ge l(n,p) \textnormal{ for all } \bx \in \{0,1\}^n, \label{eq:cha} \\
    H_{\min}^{\mu_n}(\bX |\bZ)_{W^n \times U_{\cX^n} } & \ge l(n,p), \label{eq:chr}
    \end{align}
    and
    \begin{align}
        & \lim_{n \to \infty} \mu_n = 0, \label{eq:DET0}\\
        & \lim_{n \to \infty} \frac{l(n,p)}{n} = h(p). \label{eq:DETH}
    \end{align}
\end{lemma*}

\begin{proof}
We will prove the lemma in three steps. 

\noindent \textbf{Step 1}: Set values of $\mu_n$ and $l(n,p)$.

Take $\varepsilon = n^{-1/3}, \mu_n = 8 \cdot 2^{-\sqrt[3]{n}}, l(n,p) =n \left( h(p) - c \varepsilon \right)$, where $c =\log ((1-p)/p)$ is a constant. Then equation \eqref{eq:DET0} and \eqref{eq:DETH} are obvious. 

\noindent \textbf{Step 2}: Prove inequality \eqref{eq:cha}.

By Lemma \ref{lem:unit}, we have
\begin{align} \label{eq:unit2}
    \Pr_{\bZ \sim W^n(\bx)}[\bZ \in \cT_{W, \varepsilon}^n(\bx)] = \sum_{\bz \in \cT^n_{W,\varepsilon}(\bx)} W^n(\bz|\bx) \ge 1 - \mu_n. 
\end{align}
Let $Q_{\bX\bZ}$ be a subnormalized distribution on $\cX^n \times \cZ^n$ satisfying 
% \begin{itemize}
%     \item if $\bz \in \cT_{W, \varepsilon}^n(\bx)$, then $Q_{\bZ}(\bz) = W^n(\bz|\bx)$;
%     \item if $\bz \notin \cT_{W, \varepsilon}^n(\bx)$, then $Q_{\bZ}(\bz) = 0$.
% \end{itemize}
\begin{align*}
    Q_{\bX\bZ}(\bx,\bz) = \frac{1}{2^n}Q_{\bZ|\bX}(\bz|\bx),
\end{align*}
where 
\begin{align} \label{eq:CST}
    Q_{\bZ|\bX}(\bz|\bx) =\begin{cases}
        0, &\text{if } \bz \notin \cT_{W, \varepsilon}^n(\bx)\\
        W^n(\bz|\bx), &\text{if } \bz \in \cT_{W, \varepsilon}^n(\bx)
    \end{cases}.
\end{align}
% Then for any $\bx \in \{0,1\}^n$ and $\bz \in \cT_{W, \varepsilon}^n(\bx)$,
% \begin{align}
%      -\log Q_{\bZ|\bX}(\bz | \bx) = -\log p^{\HD(\bz,\bx)} (1-p)^{n - \HD(\bz,\bx)} \ge 
% \end{align}
% \begin{align} \label{eq:eqpt}
%     2^{-n \left( H(Z|X)_{W \times P_{\bx}} + c \varepsilon \right) } \le W^n(\bz | \bx) &\le 2^{-n(H(Z|X)_{W \times P_{\bx}} - c \varepsilon )},
% \end{align}
% where $P_{\bx}$ is the empirical distribution of $\bx$.
Then for any fixed $\bx  \in \{0,1\}^n $
\begin{align}    
&{} \GTD(Q_{\bZ|\bX=\bx}, W^n(\bx)) \notag\\
={}& \frac{1}{2} \sum_{\bz \notin \cT^n_{W,\varepsilon}(\bx)}  W^n(\bz|\bx)  + \frac{1}{2} \left( 1 - \sum_{\bz \in \cT^n_{W,\varepsilon}(\bx)} W^n(\bz|\bx) \right) \notag\\ 
\le{}& \mu_n, \label{eq:tdx}
\end{align}

Hence, $\forall \bx \in \{0,1\}^n$, we have
\begin{flalign} 
   &&H_{\min}^{\mu_n}(\bZ)_{W^A(\bx)} & \ge H_{\min}(\bZ)_{Q_{\bZ|\bX=\bx}} && \llap{\text{[From  inequality \eqref{eq:tdx}]}} \notag\\
   && &= -\log \max_{\bz \in \{0,1\}^n} Q_{\bZ|\bX}(\bz|\bx) && \llap{\text{[From equation \eqref{eq:defmin}]}} \notag \\
   && &= -\log \max_{\bz \in \cT_{W, \varepsilon}^n(\bx) } W^n(\bz|\bx) && \llap{\text{[From equation \eqref{eq:CST}]}}\notag \\
   && &= -\log \max_{\bz \in \cT_{W, \varepsilon}^n(\bx) } p^{\HD(\bx,\bz)} (1-p)^{n - \HD(\bx,\bz)} && \notag\\
   && &\ge n( h(p) - \varepsilon \log \frac{1-p}{p}) = l(n, p), && \label{eq:XBS}
\end{flalign}
where the last inequality follows from equation \eqref{eq:defsct}.
%where the inequality follows from ? and the definition in equation \eqref{eq:defsm}, and the second inequality follows from inequality \eqref{eq:eqpt}.

\noindent \textbf{Step 3}: Prove inequality \eqref{eq:chr}.

From inequality \eqref{eq:tdx} we have 
\begin{align}    
\GTD(Q_{\bX\bZ}, W^n\times U_{\cX^n}) = \frac{1}{2^n}  \sum_{\bx \in \{0,1\} } \GTD(Q_{\bZ|\bX=\bx}, W^n(\bx))
\le{} & \mu_n, \label{eq:tdx2}
\end{align}

Inequality \eqref{eq:chr} is obtained as follows:
\begin{flalign*}
    && H_{\min}^{\mu_n}(\bX |B)_{W^n \times U_{\cX^n} } &\ge H_{\min}(\bX |B)_{Q_{\bX\bZ}} && \llap{\text{[From  inequality \eqref{eq:tdx2}]}} \\
    && &= -\log \sum_{\bz \in \{0,1\}^n} \max_{\bx \in \{0,1\}^n} \frac{1}{2^n} Q_{\bZ|\bX}(\bz| \bx) && \llap{\text{[From equation \eqref{eq:def-c}]}}  \\
    && & \ge l(n,p).
\end{flalign*}
The last inequality holds because $\forall \bx,\bz \in \{0,1\}^n$, $Q_{\bZ|\bX}(\bz| \bx) \le p^{n(p-\varepsilon)}(1-p)^{n(1-p+\varepsilon)}$ .
\end{proof}


\begin{thebibliography}{10}

\bibitem{blum1983coin}
M.~Blum, ``Coin flipping by telephone a protocol for solving impossible problems,'' {\em ACM SIGACT News}, vol.~15, pp.~23--27, Jan. 1983.

\bibitem{demay2013unfair}
G.~Demay and U.~Maurer, ``Unfair coin tossing,'' in {\em 2013 {{IEEE International Symposium}} on {{Information Theory}}}, pp.~1556--1560, July 2013.

\bibitem{brassard1988minimum}
G.~Brassard, D.~Chaum, and C.~Cr{\'e}peau, ``Minimum disclosure proofs of knowledge,'' {\em Journal of Computer and System Sciences}, vol.~37, pp.~156--189, Oct. 1988.

\bibitem{goldreich1991proofs}
O.~Goldreich, S.~Micali, and A.~Wigderson, ``Proofs that yield nothing but their validity or all languages in {{NP}} have zero-knowledge proof systems,'' {\em J. ACM}, vol.~38, pp.~690--728, July 1991.

\bibitem{crepeau2020commitment}
C.~Cr{\'e}peau, R.~Dowsley, and A.~C.~A. Nascimento, ``On the commitment capacity of unfair noisy channels,'' {\em IEEE Transactions on Information Theory}, vol.~66, pp.~3745--3752, June 2020.

\bibitem{ben-or2006secure}
M.~Ben-Or, C.~Cr{\'e}peau, D.~Gottesman, A.~Hassidim, and A.~Smith, ``Secure multiparty quantum computation with (only) a strict honest majority,'' in {\em 2006 47th {{Annual IEEE Symposium}} on {{Foundations}} of {{Computer Science}} ({{FOCS}}'06)}, pp.~249--260, Oct. 2006.

\bibitem{dupuis2010secure}
F.~Dupuis, J.~B. Nielsen, and L.~Salvail, ``Secure two-party quantum evaluation of unitaries against specious adversaries,'' in {\em Advances in {{Cryptology}} -- {{CRYPTO}} 2010} (T.~Rabin, ed.), Lecture {{Notes}} in {{Computer Science}}, (Berlin, Heidelberg), pp.~685--706, Springer, 2010.

\bibitem{goldreich2019how}
O.~Goldreich, S.~Micali, and A.~Wigderson, ``How to play any mental game, or a completeness theorem for protocols with honest majority,'' in {\em Providing {{Sound Foundations}} for {{Cryptography}}: {{On}} the {{Work}} of {{Shafi Goldwasser}} and {{Silvio Micali}}}, pp.~307--328, New York, NY, USA: Association for Computing Machinery, Oct. 2019.

\bibitem{lo1997quantum}
H.-K. Lo and H.~F. Chau, ``Is quantum bit commitment really possible?,'' {\em Physical Review Letters}, vol.~78, pp.~3410--3413, Apr. 1997.

\bibitem{mayers1997unconditionally}
D.~Mayers, ``Unconditionally secure quantum bit commitment is impossible,'' {\em Physical Review Letters}, vol.~78, pp.~3414--3417, Apr. 1997.

\bibitem{lo1998why}
H.-K. Lo and H.~F. Chau, ``Why quantum bit commitment and ideal quantum coin tossing are impossible,'' {\em Physica D: Nonlinear Phenomena}, vol.~120, pp.~177--187, Sept. 1998.

\bibitem{winkler2011commitment}
S.~Winkler, M.~Tomamichel, S.~Hengl, and R.~Renner, ``Impossibility of growing quantum bit commitments,'' {\em Phys. Rev. Lett.}, vol.~107, p.~090502, Aug 2011.

\bibitem{kent1999relativistic}
A.~Kent, ``Unconditionally secure bit commitment,'' {\em Phys. Rev. Lett.}, vol.~83, pp.~1447--1450, Aug 1999.

\bibitem{croke2012relativistic}
S.~Croke and A.~Kent, ``Security details for bit commitment by transmitting measurement outcomes,'' {\em Phys. Rev. A}, vol.~86, p.~052309, Nov 2012.

\bibitem{kaniewski2013relativistic}
J.~Kaniewski, M.~Tomamichel, E.~Hänggi, and S.~Wehner, ``Secure bit commitment from relativistic constraints,'' {\em IEEE Transactions on Information Theory}, vol.~59, no.~7, pp.~4687--4699, 2013.

\bibitem{lunghi2013relativistic}
T.~Lunghi, J.~Kaniewski, F.~Bussi\`eres, R.~Houlmann, M.~Tomamichel, A.~Kent, N.~Gisin, S.~Wehner, and H.~Zbinden, ``Experimental bit commitment based on quantum communication and special relativity,'' {\em Phys. Rev. Lett.}, vol.~111, p.~180504, Nov 2013.

\bibitem{lunghi2015relativistic}
T.~Lunghi, J.~Kaniewski, F.~Bussi\`eres, R.~Houlmann, M.~Tomamichel, S.~Wehner, and H.~Zbinden, ``Practical relativistic bit commitment,'' {\em Phys. Rev. Lett.}, vol.~115, p.~030502, Jul 2015.

\bibitem{damgard2008cryptography}
I.~B. Damg{\AA}rd, S.~Fehr, L.~Salvail, and C.~Schaffner, ``Cryptography in the bounded-quantum-storage model,'' {\em SIAM Journal on Computing}, vol.~37, pp.~1865--1890, Jan. 2008.

\bibitem{wehner2008composable}
S.~Wehner and J.~Wullschleger, ``Composable security in the bounded-quantum-storage model,'' in {\em Automata, {{Languages}} and {{Programming}}} (L.~Aceto, I.~Damg{\aa}rd, L.~A. Goldberg, M.~M. Halld{\'o}rsson, A.~Ing{\'o}lfsd{\'o}ttir, and I.~Walukiewicz, eds.), Lecture {{Notes}} in {{Computer Science}}, ({Berlin, Heidelberg}), pp.~604--615, {Springer}, 2008.

\bibitem{barhoush2023powerful}
M.~Barhoush and L.~Salvail, ``Powerful primitives in the bounded quantum storage model,'' June 2023.

\bibitem{wehner2008cryptography}
S.~Wehner, C.~Schaffner, and B.~M. Terhal, ``Cryptography from noisy storage,'' {\em Physical Review Letters}, vol.~100, p.~220502, June 2008.

\bibitem{wehner2008cryptographyb}
S.~Wehner, ``Cryptography in a quantum world,'' June 2008.

\bibitem{schaffnerchristian2009robust}
S.~Christian, T.~Barbara, and W.~Stephanie, ``Robust cryptography in the noisy-quantum-storage model,'' {\em Quantum Information \& Computation}, Nov. 2009.

\bibitem{konig2012unconditional}
R.~K{\"o}nig, S.~Wehner, and J.~Wullschleger, ``Unconditional security from noisy quantum storage,'' {\em IEEE Transactions on Information Theory}, vol.~58, pp.~1962--1984, Mar. 2012.

\bibitem{wyner1975wiretap}
A.~D. Wyner, ``The wire-tap channel,'' {\em Bell System Technical Journal}, vol.~54, no.~8, pp.~1355--1387, 1975.

\bibitem{csiszar1978broadcast}
I.~Csisz\'{a}r and J.~K\"{o}rner, ``Broadcast channels with confidential messages,'' {\em IEEE Transactions on Information Theory}, vol.~24, no.~3, pp.~339--348, 1978.

\bibitem{crepeau1997efficient}
C.~Cr{\'{e}}peau, ``Efficient {Cryptographic} {Protocols} {Based} on {Noisy} {Channels},'' in {\em Advances in {Cryptology} — {EUROCRYPT} ’97} (W.~Fumy, ed.), Lecture {Notes} in {Computer} {Science}, (Berlin, Heidelberg), pp.~306--317, Springer, 1997.

\bibitem{winter2003commitment}
A.~Winter, A.~C.~A. Nascimento, and H.~Imai, ``Commitment capacity of discrete memoryless channels,'' in {\em Cryptography and {Coding}} (K.~G. Paterson, ed.), Lecture {Notes} in {Computer} {Science}, (Berlin, Heidelberg), pp.~35--51, Springer, 2003.

\bibitem{crepeau2005efficient}
C.~Cr{\'e}peau, K.~Morozov, and S.~Wolf, ``Efficient unconditional oblivious transfer from almost any noisy channel,'' in {\em Security in {{Communication Networks}}} (C.~Blundo and S.~Cimato, eds.), Lecture {{Notes}} in {{Computer Science}}, (Berlin, Heidelberg), pp.~47--59, Springer, 2005.

\bibitem{hayashi2022commitment}
M.~Hayashi and N.~A. Warsi, ``Commitment capacity of classical-quantum channels,'' in {\em 2022 {{IEEE International Symposium}} on {{Information Theory}} ({{ISIT}})}, pp.~1058--1063, June 2022.

\bibitem{hayashi2023commitment}
M.~Hayashi and N.~A. Warsi, ``Commitment capacity of classical-quantum channels,'' {\em IEEE Transactions on Information Theory}, vol.~69, pp.~5083--5099, Aug. 2023.

\bibitem{ishai2011constantrate}
Y.~Ishai, E.~Kushilevitz, R.~Ostrovsky, M.~Prabhakaran, A.~Sahai, and J.~Wullschleger, ``Constant-rate oblivious transfer from noisy channels,'' in {\em Advances in {{Cryptology}} -- {{CRYPTO}} 2011} (P.~Rogaway, ed.), Lecture {{Notes}} in {{Computer Science}}, (Berlin, Heidelberg), pp.~667--684, Springer, 2011.

\bibitem{dowsley2017oblivious}
R.~Dowsley and A.~C.~A. Nascimento, ``On the oblivious transfer capacity of generalized erasure channels against malicious adversaries: The case of low erasure probability,'' {\em IEEE Transactions on Information Theory}, vol.~63, pp.~6819--6826, Oct. 2017.

\bibitem{damgard1999im}
I.~Damg{\aa}rd, J.~Kilian, and L.~Salvail, ``On the (im)possibility of basing oblivious transfer and bit commitment on weakened security assumptions,'' in {\em Advances in {{Cryptology}} --- {{EUROCRYPT}} '99} (J.~Stern, ed.), Lecture {{Notes}} in {{Computer Science}}, (Berlin, Heidelberg), pp.~56--73, Springer, 1999.

\bibitem{cascudo2016oblivious}
I.~Cascudo, I.~Damg{\aa}rd, F.~Lacerda, and S.~Ranellucci, ``Oblivious transfer from any non-trivial elastic noisy channel via secret key agreement,'' in {\em Theory of {{Cryptography}}} (M.~Hirt and A.~Smith, eds.), (Berlin, Heidelberg), pp.~204--234, Springer, 2016.

\bibitem{khurana2016secure}
D.~Khurana, H.~K. Maji, and A.~Sahai, ``Secure computation from elastic noisy channels,'' in {\em Advances in {{Cryptology}} -- {{EUROCRYPT}} 2016} (M.~Fischlin and J.-S. Coron, eds.), (Berlin, Heidelberg), pp.~184--212, Springer, 2016.

\bibitem{imai2006efficient}
H.~Imai, K.~Morozov, A.~C. A.~Nascimento, and A.~Winter, ``Efficient protocols achieving the commitment capacity of noisy correlations,'' in {\em 2006 {{IEEE International Symposium}} on {{Information Theory}}}, pp.~1432--1436, July 2006.

\bibitem{tomamichel2015quantum}
M.~Tomamichel, {\em Quantum information processing with finite resources: mathematical foundations}.
\newblock Springer Cham, 2015.

\bibitem{gutoski2007toward}
G.~Gutoski and J.~Watrous, ``Toward a general theory of quantum games,'' in {\em Proceedings of the Thirty-Ninth Annual ACM Symposium on Theory of Computing}, STOC '07, (New York, NY, USA), p.~565–574, Association for Computing Machinery, 2007.

\bibitem{chiribella2009theoretical}
G.~Chiribella, G.~M. D’Ariano, and P.~Perinotti, ``Theoretical framework for quantum networks,'' {\em Physical Review A}, vol.~80, p.~022339, Aug. 2009.
\newblock Publisher: American Physical Society.

\bibitem{renner2008security}
R.~Renner, ``Security of quantum key distribution,'' {\em International Journal of Quantum Information}, vol.~06, no.~01, pp.~1--127, 2008.

\bibitem{bennett1995generalized}
C.~H. Bennett, G.~Brassard, C.~Cr{\'e}peau, and U.~M. Maurer, ``Generalized privacy amplification,'' {\em IEEE Transactions on Information Theory}, vol.~41, no.~6, pp.~1915--1923, 1995.

\bibitem{hastad1999pseudorandom}
J.~H{\AA}stad, R.~Impagliazzo, L.~A. Levin, and M.~Luby, ``A pseudorandom generator from any one-way function,'' {\em SIAM Journal on Computing}, vol.~28, pp.~1364--1396, Jan. 1999.

\bibitem{imai2004rates}
H.~Imai, J.~{Muller-Quade}, A.~Nascimento, and A.~Winter, ``Rates for bit commitment and coin tossing from noisy correlation,'' in {\em International {{Symposium on Information Theory}}, 2004. {{ISIT}} 2004. {{Proceedings}}.}, pp.~45--, June 2004.

\bibitem{damgaard2007tight}
I.~B. Damg{\aa}rd, S.~Fehr, R.~Renner, L.~Salvail, and C.~Schaffner, ``A tight high-order entropic quantum uncertainty relation with applications,'' Aug. 2007.

\bibitem{tomamichel2011uncertainty}
M.~Tomamichel and R.~Renner, ``Uncertainty relation for smooth entropies,'' {\em Physical Review Letters}, vol.~106, p.~110506, Mar. 2011.

\end{thebibliography}
\end{document}